\documentclass{article}[11pt]  
\usepackage[margin=1.in]{geometry}
\usepackage{amsmath}
\usepackage{amssymb}
\usepackage{amsthm}
\usepackage{ifpdf}
\usepackage{wrapfig}
\usepackage{fullpage}
\usepackage{cite}
\usepackage{float}
\usepackage{graphicx}
\usepackage{subcaption}
\usepackage{comment}
\usepackage{bbding}
\usepackage{multirow}
\usepackage{enumitem}
\usepackage{algorithm}
\usepackage{algpseudocode}

\newtheorem{theorem}{Theorem}[section]
\newtheorem{lemma}[theorem]{Lemma}
\newtheorem{corollary}[theorem]{Corollary}
\newtheorem{definition}[theorem]{Definition}

\expandafter\let\expandafter\oldproof\csname\string\proof\endcsname

\renewenvironment{proof}[1][\proofname]{%
  \oldproof[#1]%
}{}  

\newcommand{\BO}{\mathcal{O}}
\newcommand{\probterm}[2]{\texttt{PROB}^{#2}_{#1}}
\newcommand{\prodset}[1]{\texttt{PROD}_{#1}}
\newcommand{\termset}[1]{\texttt{TERM}_{#1}}

\begin{document}

\title{Concentration Independent Random Number Generation in Tile Self-Assembly\footnote{This research was supported in part by National Science Foundation Grants CCF-1117672 and CCF-1555626.}\footnotetext{A preliminary version of some of these results appeared in~\cite{cameron2015flipping}  and \cite{Cameron2015arxiv}.}}

\author{
Cameron Chalk\footnotemark[1]
\and
Bin Fu\footnotemark[1]
\and
Eric Martinez\footnotemark[1]
\and
Robert Schweller\footnotemark[1]
\and
Tim Wylie\footnotemark[1]
}

\date{}
\clearpage\maketitle
\thispagestyle{empty}

\vspace*{-.5cm}
\begin{center}
Department of Computer Science\\The University of Texas - Rio Grande Valley\\Edinburg, TX, 78539-2999 \\
{\normalfont \{cameron.chalk01, bin.fu, eric.m.martinez02, robert.schweller, timothy.wylie\}@utrgv.edu}
\end{center}


\begin{abstract}
In this paper we introduce the \emph{robust random number generation} problem where the goal is to design an abstract tile assembly system (aTAM system) whose terminal assemblies can be split into $n$ partitions such that a resulting assembly of the system lies within each partition with probability 1/$n$, regardless of the relative concentration assignment of the tile types in the system.  First, we show this is possible for $n=2$ (a \emph{robust fair coin flip}) within the aTAM, and that such systems guarantee a worst case $\BO(1)$ space usage.  We accompany our primary construction with variants that show trade-offs in space complexity, initial seed size, temperature, tile complexity, bias, and extensibility, and also prove some negative results.  As an application, we combine our coin-flip system with a result of Chandran, Gopalkrishnan, and Reif to show that for any positive integer $n$, there exists a $\BO(\log n)$ tile system that assembles a constant-width linear assembly of expected length $n$ for any concentration assignment.  We then extend our robust fair coin flip result to solve the problem of robust random number generation in the aTAM for all $n$. Two variants of robust random bit generation solutions are presented: an unbounded space solution and a bounded space solution which incurs a small bias.  Further, we consider the harder scenario where tile concentrations change arbitrarily at each assembly step and show that while this is not possible in the aTAM, the problem can be solved by exotic tile assembly models from the literature.
\end{abstract}

\newpage
\setcounter{page}{1}

\section{Introduction} \label{sec:introduction}

\emph{Self-assembly} is the process by which local interactivity among unorganized, autonomous units results in their amalgamation into more complex compounds.  One of the premiere models for studying the theoretical possibilities of self-assembly is the \emph{abstract tile assembly model} (aTAM)~\cite{Winf98} in which system monomers are 4-sided tiles (inspired by Wang tiles~\cite{Wang61}) that attach to a growing seed assembly when matching glues present a sufficient bonding strength.  The motivation for studying the aTAM stems from the feasibility of a nanoscale DNA implementation~\cite{ContantineThesis}, along with the universal computational power of the model~\cite{RotWin00}, which permits many features including \emph{algorithmic} self-assembly of general shapes~\cite{SolWin07}, and more~\cite{Doty-2012a,nacoTileSurvey}.


\begin{table}[t]
\captionsetup{font=bf}
\centering
\begin{tabular}{lc}
   \renewcommand{\arraystretch}{1.2}
    \begin{tabular}[t]{cccccc}

        \multicolumn{6}{c}{\textbf{Robust Coin Flip in the aTAM}}                                                                                                                                                                                                   \\ \hline
        \multicolumn{1}{|c|}{\textbf{Space}}             & \multicolumn{1}{c|}{\textbf{Bias}}            & \multicolumn{1}{c|}{$\tau$}            & \multicolumn{1}{c|}{$|\sigma|$} & \multicolumn{1}{c|}{\textbf{$k$-ext}}   & \multicolumn{1}{c|}{\textbf{Theorem}} \\ \hline
        \multicolumn{1}{|c|}{$\BO(1)$}          & \multicolumn{1}{c|}{-}               & \multicolumn{1}{c|}{1}                 & \multicolumn{1}{c|}{7}          & \multicolumn{1}{c|}{2}       & \multicolumn{1}{c|}{\ref{2extpositive}}       \\ \hline
        \multicolumn{1}{|c|}{$\BO(1)$}          & \multicolumn{1}{c|}{-}               & \multicolumn{1}{c|}{1}                 & \multicolumn{1}{c|}{1}          & \multicolumn{1}{c|}{2}       & \multicolumn{1}{c|}{\ref{2extt1}}       \\ \hline
        \multicolumn{1}{|c|}{unbounded}         & \multicolumn{1}{c|}{-}               & \multicolumn{1}{c|}{2}                 & \multicolumn{1}{c|}{1}          & \multicolumn{1}{c|}{1}       & \multicolumn{1}{c|}{\ref{1-ext positive}}       \\ \hline
        \multicolumn{1}{|c|}{$s$}                 & \multicolumn{1}{c|}{$\le p^{(s/10)}$} & \multicolumn{1}{c|}{2}             & \multicolumn{1}{c|}{1}          & \multicolumn{1}{c|}{1}       & \multicolumn{1}{c|}{\ref{1-ext approximate positive}}       \\ \hline \\
        \end{tabular}
&
    
     \vspace*{.15cm}
    \begin{tabular}[t]{cccc}
        \multicolumn{4}{c}{\textbf{Unstable Concentrations (Thm. \ref{thm:secondary_models})}}                           \\ \hline
        \multicolumn{1}{|c|}{\textbf{Model}} & \multicolumn{1}{|c|}{\textbf{Space}}                       & \multicolumn{1}{c|}{$\tau$}            & \multicolumn{1}{c|}{$|\sigma|$}           \\ \hline
        \multicolumn{1}{|c|}{neg-aTAM} & \multicolumn{1}{|c|}{$\BO(1)$}                             & \multicolumn{1}{c|}{1}          & \multicolumn{1}{c|}{2}               \\ \hline
        \multicolumn{1}{|c|}{neg-hTAM} & \multicolumn{1}{|c|}{$\BO(1)$}                        & \multicolumn{1}{c|}{1}          & \multicolumn{1}{c|}{1}        \\ \hline
        \multicolumn{1}{|c|}{polyTAM} & \multicolumn{1}{|c|}{$\BO(1)$}                          & \multicolumn{1}{c|}{2}          & \multicolumn{1}{c|}{3}        \\ \hline
        \multicolumn{1}{|c|}{GTAM} & \multicolumn{1}{|c|}{$\BO(1)$}                            & \multicolumn{1}{c|}{1}          & \multicolumn{1}{c|}{2}            \\ \hline
    \end{tabular}
    
    \\

    \begin{tabular}[t]{cccccc}
        \multicolumn{6}{c}{\textbf{General Random Number Generation}}     \\
        \hline
        \multicolumn{1}{|c|}{\textbf{Space}}             & \multicolumn{1}{c|}{\textbf{Bias}}            & \multicolumn{1}{c|}{$\tau$}            & \multicolumn{1}{c|}{$|\sigma|$} & \multicolumn{1}{c|}{\textbf{$k$-ext}}   & \multicolumn{1}{c|}{\textbf{Theorem}} \\ \hline
        \multicolumn{1}{|c|}{unbounded}             & \multicolumn{1}{c|}{-}            & \multicolumn{1}{c|}{2}            & \multicolumn{1}{c|}{1} & \multicolumn{1}{c|}{2}   & \multicolumn{1}{c|}{\ref{thm:1nrandom}} \\
        \hline

         \multicolumn{1}{|c|}{$s$}             & \multicolumn{1}{c|}{$\leq \frac{1}{2^{\Theta \left (s/\log n \right )}}$}            & \multicolumn{1}{c|}{2}            & \multicolumn{1}{c|}{1} & \multicolumn{1}{c|}{2}   & \multicolumn{1}{c|}{\ref{thm:1nrandomBounded}}  \\

        \hline \\
        \end{tabular}
   &


    \vspace*{.15cm}
    \begin{tabular}[t]{|c|c|c|c|c|}
        \multicolumn{5}{c}{\textbf{Robust Linear Assemblies}}     \\
        \hline
        \textbf{$|T|$}      & \textbf{Width}    & $\tau$    & $|\sigma|$    & \textbf{Theorem} \\ \hline
        $\BO(\log n)$       & 4                 & 2         & 1             & \ref{theorem:simulate}, \ref{corollary:expected_length}\\ \hline
         $\BO(\log n)$      & 6                 & 1         & 1             & \ref{theorem:simulate1}, \ref{corollary:expected_lengtht1}  \\
         \hline
        \end{tabular}
    \\ \\
\end{tabular}


\captionsetup{font=normalsize}

\caption {$\tau$ represents the temperature of the system, $|\sigma|$ is the number of tiles in the seed assembly, $|T|$ is the size of the tile system, and \emph{k}-ext denotes the extensibility of the system.  Given the largest disparity in relative tile concentration between any pair of tile types in the system for a given concentration distribution, $p$ is the larger relative concentration of the two tiles. $s$ is a space constraint, and $n$ is the range of possible values or the length of the linear assembly.}
\label{tab:results}

\end{table}


A promising new direction in self-assembly is the consideration of \emph{randomized} self-assembly systems.  In randomized self-assembly (a.k.a. nondeterministic self-assembly), assembly growth is dictated by nondeterministic, competing assembly paths yielding a probability distribution on a set of final, terminal assemblies. Through careful design of tile-sets and the relative concentration distributions of these tiles, a number of new functionalities and efficiencies have been achieved that are provably impossible without this nondeterminism.  For example, by precisely setting the concentration values of a generic set of tile species, arbitrarily complex strings of bits can be \emph{programmed} into the system to achieve a specific shape with high probability~\cite{Dot10,KaoSchS08}.  Alternately, if the concentration of the system is assumed to be fixed at a uniform distribution, randomization still provides for efficient expected growth of linear assemblies~\cite{ChaGopRei12} and low-error computation at temperature-1~\cite{CookFuSch11}.  Even in the case where concentrations are unknown, randomized self-assembly can build certain classes of shapes without error in a more efficient manner than without randomization~\cite{BryChiDotKarSekJournal}.

Motivated by the power of randomized self-assembly, along with the potential for even greater future impact, we focus on the development of the most fundamental randomization primitive: the \emph{robust} generation of a uniform random bit.  In particular, we introduce the problem of self-assembling a uniformly random bit within $\BO(1)$ space that is guaranteed to work for all possible concentration distributions.  We define a tile system to be a \emph{coin flip} system, with respect to some tile concentration distribution, if the terminal assemblies of the system can be partitioned such that each partition has exactly probability 1/2 of assembling one of its terminals.  We say a system is a \emph{robust coin flip} system if such a partition exists that guarantees 1/2 probability for all possible tile concentration distributions.  Through designing systems that flip a fair coin for all possible (adversarially chosen) concentration distributions, we achieve an intrinsically fair coin-flipping system that is robust to the experimental realities of imprecise quantity measurements.  Such fair systems may allow for increased scalability of randomized self-assembly systems in scenarios where exact concentrations of species are either unknown or intractable to predict at successive assembly stages.

\paragraph{Our results}  Our primary result is an aTAM construction that constitutes a robust fair coin flip system which completes in a guaranteed $\BO(1)$ space even at temperature one.  We apply our robust coin-flip construction to the result of Chandran, Gopalkrishnan, and Reif~\cite{ChaGopRei12} to show that for any positive integer $n$, there exists a $\BO(\log n)$ tile system that assembles a constant width-$4$ linear assembly of expected length $n$ that works for all concentration assignments.  This result is for temperature two; at temperature one it must be a width-$6$ linear assembly.  We accompany this result with a proof that such a concentration independent assembly of width-1 assemblies is not possible with fewer than $n$ tile types.  We further accompany our main coin-flip construction with variant constructions that provide trade-offs among standard aTAM metrics such as space, tile complexity, and temperature, as well as new metrics such as coin bias, and the \emph{extensibility} of the system, which is the maximum number of distinct locations a tile can be added to a single producible assembly of the system.

We utilize the coin-flip construction as a fair random bit generator for implementation of some classical random number generation algorithms.  We show that 1-extensible systems, while computationally universal, cannot robustly coin-flip in bounded space without incurring a bias, but can robustly coin-flip in bounded expected space.  We also consider the more extreme model in which concentrations may change adversarially at each assembly step.  We show that the aTAM cannot robustly coin flip in bounded space within this model, but a number of more exotic extensions of the aTAM from the literature are able to robustly coin flip in $\BO(1)$ space.  We summarize our results in Table~\ref{tab:results}.  The problem of self-assembling random bits has been considered before~\cite{RNSSA}, but their technique, and almost all randomized techniques to date, do not work when arbitrary concentrations are considered. Further,  we utilize the self-assembly of uniform random bits to implement algorithms for uniform random number generation for any $n$, one construction achieving an unbiased generator with unbounded space and the other imposing a space constraint while incurring some bias.

\paragraph{Organization} Due to the many results in the paper, we briefly outline them here.  Section \ref{sec:definitions} gives the definitions of the models and terms used throughout the paper as well as an overview of some related previous work.  In Section \ref{sec:bounded} we cover the constant space coin flipping gadget- first with a big seed and then with a single seed at temperature two.  We then use this gadget in Section \ref{sec:linear} to assemble $\BO(1) \times n$ expected length linear assembles.  Section \ref{sec:boundedt1} shows that the coin flip gadget can be built at temperature one with some extra tiles, and then we show how the expected length linear assemblies can be built with a slightly larger constant width using the temperature one gadget.

The paper then covers general random number generation in the aTAM in Section \ref{sec:rng}.  Afterwards, the paper switches focus to the limitations of different aspects of the model covering 1-extensibility in \ref{sec:1ext} and unstable concentrations in \ref{sec:unstable}.  Then in Section \ref{sec:alternative}, we show how robust fair coin flips are possible in some other models.  Finally, we conclude and give some future directions in \ref{sec:conclusion}.

\section{Definitions and Model: Tiles, Assemblies, and Tile Systems} \label{sec:definitions}

 Consider some alphabet of glue types $\Pi$.  A tile is a unit square with four edges each assigned some glue type from $\Pi$.  Further, each glue type $g \in \Pi$ has some non-negative integer strength $str(g)$.  Each tile may be assigned a finite length string label, e.g., ``black",``white",``0", or ``1". For simplicity, we assume each tile center is located at a pixel $p = (p_x, p_y) \in \mathbb{Z}^2$. For a given tile $t$, we denote the tile center of $t$ as its position. As notation, we denote the set of all tiles that constitute all translations of the tiles in a set $T$ as the set $T^*$. An \emph{assembly} is a set of tiles each assigned unique coordinates in $\mathbb{Z}^2$.  For a given assembly $\alpha$, define the \emph{bond graph} $G_\alpha$ to be the weighted graph in which each element of $\alpha$ is a vertex, and each edge weight between tiles is $str(g)$ if the tiles share an overlapping glue $g$, and 0 otherwise.  An assembly $\alpha$ is said to be \emph{$\tau$-stable} for a positive integer $\tau$ if the bond graph $G_\alpha$ has min-cut at least $\tau$, and \emph{$\tau$-unstable} otherwise. A \emph{tile system} is an ordered triple $\Gamma =(T,\sigma,\tau)$ where $T$ is a set of tiles called the \emph{tile set} (we refer to elements of $T$ as tile types), $\sigma$ is an assembly called the \emph{seed} and $\tau$ is a positive integer called the \emph{temperature}.  When considering a tile $a$ that is some translation of an element of a tile set $T$, we will use the term \emph{tile type} of $a$ to reference the element of $T$ that $a$ is a translation from. Assembly proceeds by growing from assembly $\sigma$ by any sequence of single tile attachments from $T$ as long as each tile attachment connects with strength at least $\tau$.  Formally, we define what can be built in this fashion as the set of producible assemblies:

\begin{definition}[Producibility]
    For a given tile system $\Gamma=(T,\sigma,\tau)$, the set of \textbf{producible assemblies} for system $\Gamma$, $\prodset{\Gamma}$, is defined recursively:
    \begin{itemize}
        \item (Base) $\sigma \in \prodset{\Gamma}$
        \item (Recursion) For any $A\in \prodset{\Gamma}$ and $b \in T^*$ such that $C=A\cup \{b\}$ is $\tau$-stable, then $C \in \prodset{\Gamma}$.
    \end{itemize}
\end{definition}

    As additional notation, we say $A {\rightarrow^\Gamma_1} B$ if $A$ may grow into $B$ through a single tile attachment, and we say $A {\rightarrow^\Gamma} B$ if $A$ can grow into $B$ through 0 or more tile attachments. An \textbf{assembly sequence} for a tile system $\Gamma$ is a sequence (finite or infinite) $\vec{\alpha} = \langle \alpha_{1}, \alpha_{2},\dots \rangle$ in which $\alpha_{1}=\sigma$, each $\alpha_{i+1}$ is a single-tile extension of $\alpha_{i}$, and each $\alpha_{i}$ is $\tau$-stable. The \textbf{frontier} of an assembly $\alpha$, written as $F(\omega, \Gamma)$, is a partial function that maps an assembly $\omega$ and a tile system $\Gamma$ to a set of tiles $\{t \in T^* | \omega \cup \{t\} \in \prodset{\Gamma} \land t \notin \omega \}$. We further define $\termset{\Gamma}$ to be the subset of $\prodset{\Gamma}$ consisting only of assemblies for which no further tile in $T$ may attach (i.e., the assembly has an empty frontier).

\begin{definition}[Finiteness and Space]
For a given tile assembly system $\Gamma=(T,\sigma,\tau)$, we say $\Gamma$ is \textbf{finite} iff $\forall \beta \in \prodset{\Gamma},\exists \alpha \in \termset{\Gamma} : \beta \rightarrow^\Gamma \alpha$. That is, each producible assembly has a growth path ending in a finite, terminal assembly. If $\Gamma$ is not finite, we say it is \textbf{infinite}. Define the \textbf{space of an assembly} $\alpha$ as $|\alpha|$. Let the \textbf{space of a tile assembly system} be defined as the $\max\limits_{\alpha \in \termset{\Gamma}} |\alpha|$ iff $\Gamma$ is finite. If $\Gamma$ is infinite, let \textbf{space} remain undefined.  Note that a finite system may have infinite/unbounded space.
\end{definition}

\begin{definition}[Extensibility]
    Consider a tile assembly system $\Gamma=(T,\sigma,\tau)$, and assembly $\alpha \in \prodset{\Gamma}$. We denote the set of all locations at which a tile may stably attach to $\alpha$ as $L_{\alpha}$.  More formally, $L_{\alpha} = \{p_{t} | t \in F(\alpha, \Gamma)\}$. We say a tile system $\Gamma$ is \textbf{$k$-extensible} iff $\forall \alpha \in \prodset{\Gamma}, \left\vert{L_{\alpha}}\right\vert \le k$. Informally, a tile assembly system is $k$-extensible iff at any point in the assembly process, the assembly can only grow in at most $k$ locations.
\end{definition}

\subsection{Probability in Tile Assembly}

We use the definition of probabilistic assembly presented in~\cite{BeckerRR06,ChaGopRei12,CookFuSch11,Dot10,KaoSchS08}. Let $P$ be a function denoting a \textbf{concentration distribution} over a tileset $T$ representing the concentrations of each tile type with the restrictions $\forall t \in T, P(t) > 0$ and $\sum\limits_{t \in T}P(t)=1$. For a tile $t$, we sometimes refer to $P(t)$ as the \textbf{concentration} of $t$. Using a concentration distribution, we can consider probabilities for certain events in the system. To study probabilistic assembly, we can consider the assembly process as a Markov chain where each producible assembly is a state and transitions occur with non-zero probability from assembly $A$ to each $B$ whenever $A \rightarrow^\Gamma_1 B$. For each $B$ that satisfies $A \rightarrow^\Gamma_1 B$, let $t_{A\to B}$ denote the tile in $T$ whose translation is added to $A$ to get $B$. The transition probability from $A$ to $B$ is defined to be

	\begin{equation}\label{def:trans}
	TRANS(A,B) = \dfrac{P(t_{A\to B})}{\sum\nolimits_{\{C \vert A\rightarrow^\Gamma_1 C\}} P(t_{A\to C})}
	\end{equation}

	The probability that a tile system $\Gamma$ terminally assembles an assembly $A$ is defined to be the probability that the Markov chain ends in state $A$. For each $A \in \termset{\Gamma}$, let $\probterm{\Gamma \rightarrow A}{P}$ denote the probability that $\Gamma$ terminally assembles $A$ with respect to concentration distribution $P$.

\begin{definition}[Expected Space]
    For a given finite tile system $\Gamma=(T,\sigma,\tau)$, let the expected space of $\Gamma$ relative to a concentration distribution $P$ be defined as

    \begin{equation*}
    \texttt{EXPECTEDSPACE}_\Gamma = \sum\limits_{\alpha \in \termset{\Gamma}} |\alpha|\cdot\probterm{\Gamma \rightarrow \alpha}{P}
    \end{equation*}

\end{definition}

\begin{definition}[Coin Flipping]We consider a finite tile system $\Gamma$ a \textbf{coin flip tile system with bias $b$} with respect to a concentration distribution $P$ for some $b \in \mathbb{R}^+$ iff the set of terminal assemblies in $\prodset{\Gamma}$ can be partitioned into two sets $X$ and $Y$ such that $\left|\sum\limits_{x \in X} \probterm{\Gamma \rightarrow x}{P} - \sum\limits_{y \in Y} \probterm{\Gamma \rightarrow y}{P}\right| \leq 2b$. 
A \textbf{fair coin flip tile system} is a coin flip tile system with bias $0$. 
We consider a finite tile system a \textbf{robust coin flip tile system with bias $b$} iff it is a coin flip tile system with bias $b$ for all concentration distributions; i.e. $\left|\sum\limits_{x \in X} \probterm{\Gamma \rightarrow x}{C} - \sum\limits_{y \in Y} \probterm{\Gamma \rightarrow y}{C}\right| \leq 2b$ for all concentration distributions $C$.
A \textbf{robust fair coin flip tile system} is a robust coin flip tile system with bias $0$.

\end{definition}

\section{Robust Fair Coin Flipping in the aTAM}\label{sec:bounded}

\begin{figure}[t]
    \centering
    \includegraphics[scale=2.5]{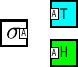}
   \caption{A non-robust fair coin flip for the uniform concentration distribution.}
    \label{simple}
\end{figure}

In this section we show systems capable of robust fair coin flips in the aTAM. Figure \ref{simple} shows a simple fair coin flip aTAM system for the uniform concentration distribution. Since the concentrations are uniform, the tiles labeled $H$ and $T$ have equal concentrations, and the probability of their attachment is equal. Then we can partition the terminal states into two sets: the terminal containing $H$ in one set and the terminal containing $T$ in the other. These partitions satisfy the definition of a fair coin flip, but only for the uniform concentration distribution. Any variation in the concentrations of the $H$ and $T$ tiles results in a coin flip with some bias.  To solve this problem for arbitrary concentration distributions, more involved techniques are required.

\begin{figure}[h]
    \centering
    \includegraphics[scale=2.7]{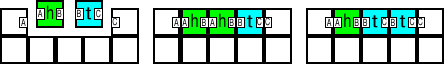}
    \caption{Shown are the $\sigma$, $h$, and $t$ tiles on the left, and the terminal states of the assembly system representing heads and tails. $A$, $B$ and $C$ glues are strength 1. Non-matching glues have 0 strength.}
    \label{bounded_tree}
\end{figure}

\begin{theorem}\label{2extpositive}
There exists a $\BO(1)$ space $2$-extensible robust fair coin flip tile system $~\Gamma~=~(T,\sigma,1)$ in the aTAM with $|\sigma|=7$.
\end{theorem}

\begin{proof}
 To show this we present a tile system $\Gamma = (T,\sigma,1)$ in which two terminal states exist and are equiprobable for all concentration distributions $P$. $|T| = 9$ and $\sigma$ contains 7 tiles. A graphical representation of $\sigma$, the two tiles $h$ and $t$, and terminal states of the assembly system is shown in Figure \ref{bounded_tree}. The system terminates nondeterministically and contains either $2$ $h$ tiles and $1$ $t$ tile or $2$ $t$ tiles and $1$ $h$ tile. The system leverages any difference in tile concentrations between $h$ and $t$ by ensuring that placement of a $t$ tile increases the probability of terminating in an assembly containing $2 h$ tiles and vice versa. Without loss of generality, assume the leftmost bottom tile in $\sigma$ sits at position $(0,0)$. We will refer to each producible assembly sans $\sigma$ by the labels of the tiles in positions $(1,1), (2,1)$ and $(3,1)$ as such: $\_\_t, h\_\_, \_ht$ and so forth. Then, the two terminal assemblies of the system are $hht$ and $htt$. We now show that $\probterm{\Gamma \rightarrow hht}{P} = \dfrac{1}{2}$ for all concentration distributions $P$. Let $c_h$ be the concentration of the tile labeled $h$ and $c_t$ be the concentration of the tile labeled $t$; then, with $TRANS$ as defined in Equation~\ref{def:trans},
\begin{align*}
    \begin{split}
    \probterm{\Gamma \rightarrow hht}{P} &= TRANS(\sigma,\_\_t) \cdotp TRANS(\_\_t,\_ht) \cdotp TRANS(\_ht,hht) \\
        &\hspace{4 mm}+ TRANS(\sigma,\_\_t) \cdotp TRANS(\_\_t,h\_t) \cdotp TRANS(h\_t,hht) \\
        &\hspace{4 mm}+ TRANS(\sigma,h\_\_) \cdotp TRANS(h\_\_,h\_t) \cdotp TRANS(h\_t,hht) \\
        &= \dfrac{c_t}{c_t + c_h} \cdotp \dfrac{c_h}{c_h + c_h} \cdotp \dfrac{c_h}{c_h}
        + \dfrac{c_t}{c_t + c_h} \cdotp \dfrac{c_h}{c_h + c_h} \cdotp \dfrac{c_h}{c_t + c_h}\\
        &\hspace{4 mm}+ \dfrac{c_h}{c_t + c_h} \cdotp \dfrac{c_t}{c_t + c_t} \cdotp \dfrac{c_h}{c_t + c_h}\\
        &= \dfrac{{c_t}^2 + 2c_tc_h + {c_h}^2}{2{c_t}^2 + 4c_tc_h + 2{c_h}^2} \\
        &= \dfrac{1}{2}.
    \end{split}
\end{align*}\hfill \qed
\end{proof}

\subsection{Extension to a Single-Seed}\label{singleseed}
A common constraint in the aTAM is that $\sigma$ contains only one tile. Thus, no seed structure must be formed prior to the self-assembly process. The construction shown in Figure \ref{seed} addresses this constraint and works in a similar fashion as the construction in Theorem~\ref{2extpositive}. Note that this system requires $\tau = 2$.

\begin{theorem}\label{2extpositivesingleseed}
There exists a $\BO(1)$ space 2-extensible robust fair coin flip tile system $\Gamma=(T,\sigma,2)$ in the aTAM with $|\sigma|=1$.
\end{theorem}

\begin{proof}
Our tile set is shown in Figure \ref{seed}. Without loss of generality, assume $\sigma$ sits at position $(0,0)$. Until the tile labeled $S$ (see Figure \ref{seed}) is placed, the assembly process is deterministic. Upon attachment of $S$, cooperative binding locations allow the attachment of tiles $h$ and $t$ nondeterministically. We denote the assemblies following the placement of $S$ similarly to the proof of Theorem~\ref{2extpositive}. We refer to assemblies containing tile $S$ by the labels of tiles in positions $(1,-1), (1,0)$ and $(2,0)$ as $\_\_t, \_\_h, \_ht$ and so forth. Reflecting the analysis shown in Theorem~\ref{2extpositive}, we have $\probterm{\Gamma \rightarrow hht}{P} = \frac{1}{2}$ for all concentration distributions $P$, which implies $\probterm{\Gamma \rightarrow htt}{P} = \frac{1}{2}$ as there are two terminal assemblies. \hfill \qed
\end{proof}

\begin{figure}[t]
    \centering
    \includegraphics[scale=2.7]{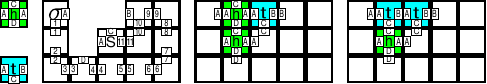}
    \caption{$T$ is shown. Our seed, labeled $\sigma$, begins a deterministic attachment process ending with the placement of the tile labeled $S$. Glues labeled $\{1, 2, 3,\dots, 11\}$ are of strength $2$. Glues labeled $\{A,B,C,D\}$ are of strength $1$, ensuring that the nondeterministic attachments of tiles $h$ and $t$ do not begin until the cooperative binding locations are opened by placement of the tile labeled $S$. The nondeterministic sequence of attachments following the placement of $S$ is similar to that of Theorem~\ref{2extpositive}.}
    \label{seed}
\end{figure}


\section{Robust Simulation of Randomized Linear Assemblies}\label{sec:linear}

As an application of the primitive shown in Theorem \ref{2extpositivesingleseed}, we show that a class of randomized linear aTAM tile assembly systems can be simulated in a concentration robust manner with a minor scale factor.

We first briefly describe a scale $(m,n)$-simulation of a given tile system, based on the block replacement schemes of~\cite{ChenGoel04}.  Consider an aTAM system $\Gamma = (T,\sigma,\tau)$ and a proposed simulator system $\Gamma'=(T',\sigma',\tau')$.  Now consider the mapping from $\termset{\Gamma}$ to $\termset{\Gamma'}$ obtained by replacing each tile in an assembly $A\in\termset{\Gamma}$ with a rectangular $m \times n$ block of tiles over $U$, according to some fixed $m \times n$ block mapping $R$.  If there exists such a mapping $M$ from $\termset{\Gamma}$ to $\termset{\Gamma'}$ that is bijective, then we say that $\Gamma'$ simulates the production of $\Gamma$ at scale factor $(m,n)$.  Further, we say that $\Gamma$ \emph{robustly simulates} $\Gamma'$ for concentration distribution $P$ if for all terminal assemblies $A \in \termset{\Gamma}$, $\probterm{\Gamma \to A}{P} = \probterm{\Gamma' \to M(A)}{C}$ for all concentration distributions $C$ over $T'$, i.e., $\Gamma'$ produces terminal assemblies with probability independent of concentration assignment, and with exactly the same probability distribution as the concentration dependent system it simulates.

We now define a class of linear assembly systems for which we can construct robust, concentration independent simulations.

\begin{definition}[Unidirectional two-choice linear assembly systems]
A tile system $\Gamma$ is a \emph{unidirectional two-choice linear assembly system} iff:
  \begin{enumerate}
  \item $\Gamma$ is $1$-extensible,
  \item $\forall \alpha \in \prodset{\Gamma}, |F(\alpha,\Gamma)| \leq 2$,
  \item $\forall \beta \in \prodset{\Gamma}, \beta$ is a $1 \times n$ line for some $n \in \mathbb{N}$.
  \end{enumerate}
\end{definition}

\begin{theorem}\label{theorem:simulate}
For any unidirectional two-choice linear assembly system $\Gamma = (T,\sigma,2)$ in the aTAM, there is an aTAM system $\Gamma_s = (T',\sigma',\tau')$ that robustly simulates $\Gamma$ for the uniform concentration distribution at scale factor $5\times 4$ with $|T'| = 20|T|$.
\end{theorem}

\begin{figure}[t!]
    \centering
    \includegraphics[width=.75\textwidth]{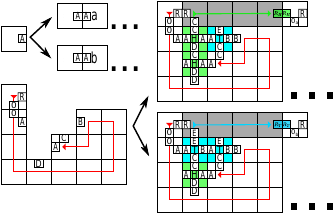}
    \caption{A simulation of one non-deterministic linear tile attachment. Each non-determinstic attachment will require a $5 \times 4$ robust coin flip gadget shown in Fig. \ref{seed}. The assembly continues after simulating an attachment by building another $5 \times 4$ robust coin flip gadget, building a deterministic $5 \times 4$ block, or terminating.}
    \label{simulation}
\end{figure}

\begin{proof} \label{app:simulate}
Let $\Gamma = (T, \sigma, \tau)$ be a unidirectional two-choice linear assembly system. Define an \emph{undecided assembly} to be any assembly $\alpha \in \prodset{\Gamma}$ such that $|F(\alpha,\Gamma)| = 2$. For each undecided assembly, we will construct a gadget utilizing the technique in Theorem~\ref{2extpositivesingleseed}. We call the two tiles of an undecided assembly's frontier $h$ and $t$. Consider $\alpha_h = \alpha \cup h$ and $\alpha_t = \alpha \cup t$. We simulate $\Gamma$ in reference to a uniform concentration distribution, so $\alpha$ transitions to $\alpha_h$ with probability $1/2$ and to $\alpha_t$ with probability $1/2$. Figure \ref{simulation} shows an example of utilizing a $5\times4$ gadget in $\Gamma_s$ to simulate the transition from $\alpha$ to $\alpha_h$ or $\alpha_t$. By application of Theorem \ref{2extpositivesingleseed}, the gadget will grow into one of two possible states with probability $1/2$ for any concentration distribution. By chaining the gadgets together we can robustly simulate the nondeterministic attachments in $\Gamma$. Each tile from the original construction is simulated by a $5\times 4$ block of tiles, and therefore $|T'| = 20|T|$. \hfill \qed
\end{proof}

As a corollary to Theorem~\ref{theorem:simulate}, we can create a tile system to build an expected length $n$ assembly for all concentration distributions with $\BO(\log n)$ tile complexity. First, we will prove that there is no aTAM tile system which generates linear (width-$1$) assemblies of expected length $n$ for all concentration distributions (\cite{ChaGopRei12} showed that this is possible for the uniform concentration distribution).

\begin{theorem}\label{thm:expected_length_neg}
There does not exist an aTAM tile system which generates an assembly of width-$1$ and expected length $n$ for all concentration distributions with
less than $n$ tile complexity.
\end{theorem}



\begin{proof}
Towards a contradiction, assume a self-assembly system can generate a linear assembly with expected length $n$ and uses at most $k < n$ tiles where $n \ge 2$. There is at least one assembly that is of length at least $n$. For all possible assemblies of length $n$, since $k < n$, they must have the form $t_1\cdots t_{i-1}t_i\cdots t_mt_i\cdots t_z$ with the first cycle $t_i\cdots t_mt_i$ for some $i,m$, which may be different for each assembly.  Define the ordering on the pair $(m,i)$ between all assemblies as $(m_1,i_1)<(m_2,i_2)$ iff $(m_1<m_2)$ or ($m_1=m_2$ and $i_1<i_2$).
Let $S = t_1\cdots t_{i-1}t_i\cdots t_mt_i\cdots t_z$ be the assembly with the minimal pair $(m,i)$ under our ordering, where $t_i \cdots t_mt_i$ is the first cycle that appears in $S$.

Since $(m,i)$ is the minimal pair for the choice of $S$, it is impossible
that the system generates $t_1\cdots t_{j-1}t_j\cdots
t_{m'}t_j\cdots $ for some $j<i$ and $m'\le m$. For $u< m$, tile
$t_k$ with $k<u$ cannot be attached to $t_1\cdots t_{u}$ to form
$t_1\cdots t_{u}t_k$. Otherwise, the system could generate $t_1\cdots t_k\cdots t_ut_k$.

We define the
concentration of the types of tiles as follows: Let $c_1 = 1, c_j = c_{j-1}/n^{100}$ for $j = 1,\cdots k$. The concentration of each type $t_i$ is ${c_i\over c_1+c_2+\cdots +c_k}$. Consider the assembly $t_1\cdots t_{j}$ with $j<m$, for those tiles $t_u$ that can be attached to the assembly, it must be the case that $u>j$. Therefore, the probability that $t_{j+1}$ is attached to it is at least
${c_{j+1}\over c_{j+1}+\sum_{u>j+1}c_u}\ge {1\over 1+n\cdot {1\over
n^{100}}}={1\over 1+ {1\over n^{99}}}$.

Consider the assembly $t_1\cdots t_{i-1}t_i\cdots t_m$ and the tile
$t_u$ with the smallest $u$ that can be attached. It must be the case that
$u=i$. Otherwise, it violates the condition that $(m,i)$ is the
minimal pair for $S$. Therefore, the probability that $t_{i}$ is attached
to it is at least ${c_i\over c_i+\sum_{u>i}c_u}\ge {1\over 1+n\cdot
{1\over n^{100}}}={1\over 1+ {1\over n^{99}}}$.

The concentration assignments ensure that with a probability of at least $\left ( {1 \over 1+ {1\over n^{99}}}\right )^{n}$ , the assembly $t_1\cdots t_{i-1}t_i\cdots t_mt_i$, or one at least as long, will be generated. This means that with a  probability of at least $\left ({1 \over 1+ {1\over n^{99}}}\right )^{n^3}> 0.9$, an assembly at least as long as $t_1\cdots t_{i-1}(t_i\cdots t_mt_i)^{n^2}$ will be generated, which has length at least $n^2$. This contradicts that the expected length is $n$. \hfill \qed
\end{proof}

We now contrast the width-1 impossibility result of Theorem~\ref{thm:expected_length_neg} with a result showing that width-4 linear assemblies do allow for efficient growth to expected length $n$ in a concentration independent manner.  To achieve this, we apply Theorem~\ref{theorem:simulate} to the unidirectional two-choice linear assembly system presented in~\cite{ChaGopRei12}, which yields the following result.

\begin{corollary}\label{corollary:expected_length}
There exists an aTAM tile system $\Gamma = (T,\sigma,\tau)$ which terminates in a width-$4$ expected length $n$ assembly for all concentration distributions. $|T| = \BO(\log n)$.
\end{corollary}

\begin{proof} \label{app:coro_expected_length}
 Let $m$ be $\left \lfloor{\dfrac{n}{5}}\right \rfloor$. Consider $\Gamma = (T,\sigma,\tau)$ to be a robust simulation at scale factor $5\times 4$ of a unidirectional two-choice linear assembly system that terminates in an expected length $m$ linear assembly using $\mathcal{O}(\log m)$ tile types. Note that such a unidirectional two-choice linear assembly system exists as shown in~\cite{ChaGopRei12} and can be robustly simulated as shown by Theorem \ref{theorem:simulate}. If $5m = n$, then $\Gamma$ terminates in an expected length $n$ assembly with width-$4$; otherwise, we add $n \mod 5$ length deterministically. Since our scale factor is constant, $|T| = \mathcal{O}(\log n)$. \hfill \qed
\end{proof}




\section{Robust Fair Coin Flipping at Temperature 1}\label{sec:boundedt1}

The previous construction for robust fair coin flipping used temperature two so that the gadget could use cooperative binding for simplicity.  However, here we show that it can be done in the aTAM at temperature one if after the coin flip the geometry ensures that the gadget is contained.  In order to do this, we use three additional tiles which relay the message and give either an $h$ or $t$ glue as the output.  Figure \ref{fig:boundedt1} shows the tile set that flips the coin and the build order of the gadget that must place its pieces in order to work correctly.  The basic construction is the same as before.

\begin{theorem}\label{2extt1}
There exists a $\BO(1)$ space 2-extensible robust fair coin flip tile system $\Gamma=(T,\sigma,\tau)$ in the aTAM with $|\sigma|=1$ and $\tau=1$.
\end{theorem}

\begin{figure}[t]
    \centering
    \includegraphics[width=.75\textwidth]{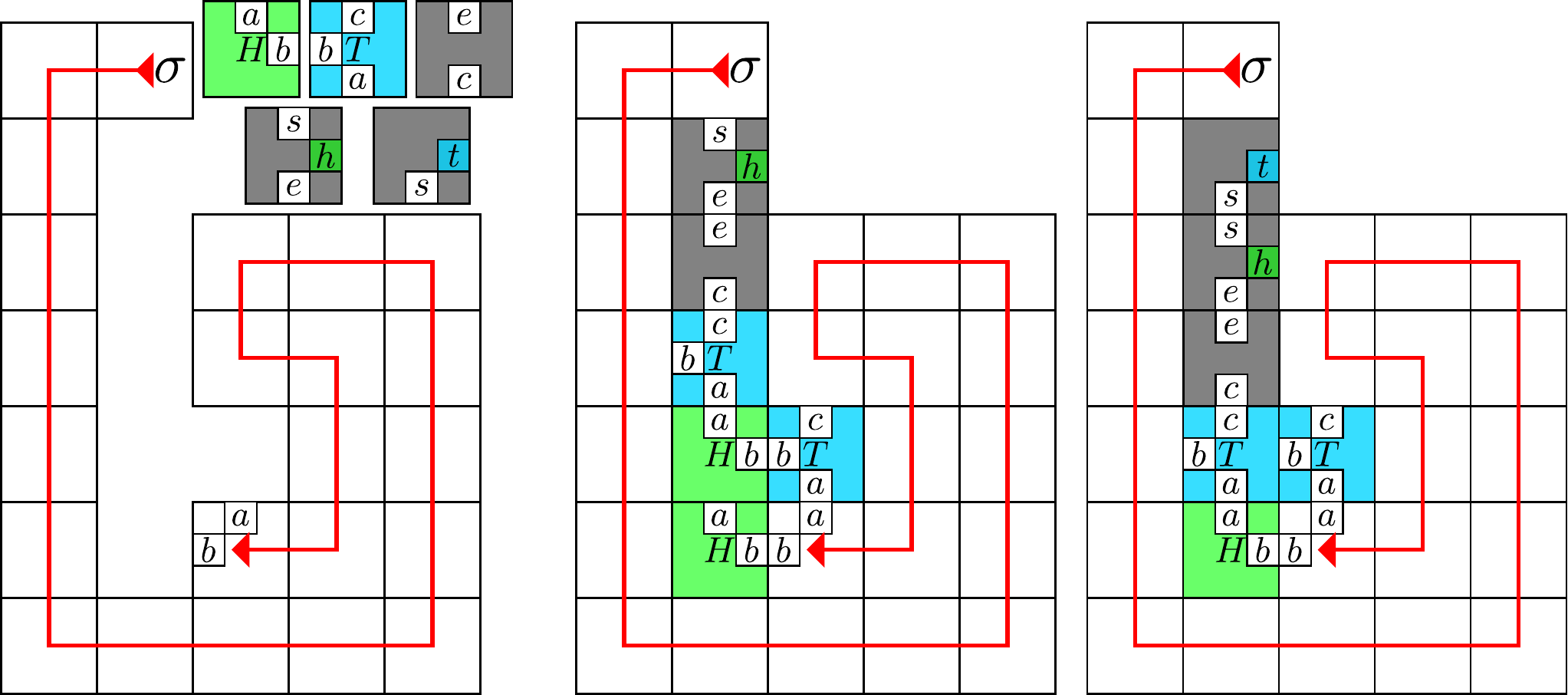}
    \caption{A robust fair coin flip gadget at temperature one. The image on the left shows the deterministic gadget and the order in which the tiles attach until it opens up a spot of the coin flip to occur.  The coin flipping tiles are shown above.  The middle figure shows a head flip which exposes the $h$ glue on top, and the right figure is the tail flip which exposes a $t$ glue on top. }
    \label{fig:boundedt1}
\end{figure}

Further, we can do expected length linear assemblies at temperature one (similar to Section \ref{sec:linear}), but due to the geometric constraints of the gadgets, the needed constant width is six.  We use the gadget shown in Figure \ref{fig:boundedt1}, but we turn it on its side in order to decrease the width (from eight to six).  This is shown in Figure \ref{fig:lint1}.  The original construction of expected length linear assemblies in \cite{ChenGoel04} only relied on $\tau=1$, and thus our gadget is able to convert their system to be concentration independent while keeping the temperature constraint and only requiring a constant width of six.

\begin{theorem}\label{theorem:simulate1}
For any unidirectional two-choice linear assembly system $\Gamma = (T,\sigma,1)$ in the aTAM, there is an aTAM system $\Gamma_s = (T',\sigma',\tau')$ that robustly simulates $\Gamma$ for the uniform concentration distribution at scale factor $9 \times 6$ with $|T'| = \BO(|T|)$.
\end{theorem}

\begin{proof}
This follows the same argument as Theorem \ref{theorem:simulate}.
\hfill \qed
\end{proof}

\begin{corollary}\label{corollary:expected_lengtht1}
There exists an aTAM tile system $\Gamma = (T,\sigma,\tau)$ which terminates in a width-$6$ expected length $n$ assembly for all concentration distributions at $\tau = 1$. $|T| = \BO(\log n)$.
\end{corollary}

\begin{figure}[t]
    \centering
    \includegraphics[width=1.\textwidth]{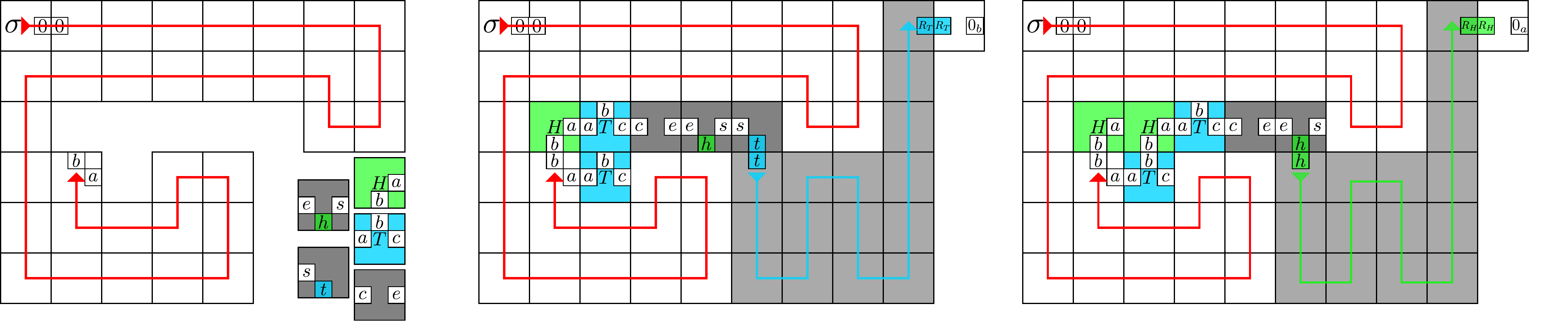}
    \caption{The temperature one coin flip gadget modified for the linear assemblies.  The left shows the gadgets and coin flip tiles.  The middle shows a tails being flipped and how the block fills in the remaining space and propogates the $t$ glue up to the top for the next gadget to attach.  The right image shows the same process when a head is flipped.}
    \label{fig:lint1}
\end{figure}

\section{Robust Random Number Generation in the aTAM}\label{sec:rng}
A natural direction following the robust fair coin-flip problem is \emph{robust random number generation} in the aTAM. The robust coin-flip solutions for the aTAM allow implementation of robust aTAM algorithms utilizing random bit generation. A more useful primitive, then, is the generation of random numbers within a given range. A generalization of our coin-flip problem definition is used to consider random number generation by aTAM systems.

\begin{definition}[Random Number Generation]We consider a finite tile system $\Gamma$ a \textbf{random $n$ generator} with respect to a concentration distribution $P$ iff the set of terminal assemblies in $\prodset{\Gamma}$ is partitionable into $n$ sets $X_1, X_2,...,X_n$ such that $\sum\limits_{x \in X_i} \probterm{\Gamma \rightarrow x}{P} = \sum\limits_{y \in X_j} \probterm{\Gamma \rightarrow y}{P}$ for $1\leq i,j \leq n$. We consider a finite tile system $\Gamma$ a \textbf{robust random $n$ generator with bias $b$} iff the set of terminal assemblies in $\prodset{\Gamma}$ is partitionable into $n$ sets $X_1, X_2,...,X_n$ such that $$\max \left \lvert \sum\limits_{x\in X_i}\probterm{\Gamma \rightarrow x}{C} - \sum\limits_{y\in X_j}\probterm{\Gamma \rightarrow y}{C} \right \rvert\ \leq b$$
over all $X_i,X_j \in X_1,\dots,X_n$ and $1\leq i,j \leq n$ and all concentration distributions $C$. A \textbf{fair robust random $n$ generator} is a robust random $n$ generator with bias $0$.
\end{definition}

\subsection{Robust Random Generator for Powers of Two}
We start off with a corollary from Theorem~\ref{2extpositivesingleseed} to solve a class of the robust random number generation problems involving powers of two.

\begin{corollary}\label{cor:2s}
For $m\in \mathbb{Z^+}$, there exists a $\BO(m)$ space $2$-extensible fair robust random $2^m$ generator $\Gamma=(T,\sigma,2)$ in the aTAM with $|\sigma| = 1$ and $|T|=\BO(\log n)$.
\end{corollary}
\begin{proof}
This corollary follows from Theorem~\ref{2extpositivesingleseed}. The single seed robust fair coin-flip system is used as a bit generator. By assembling $m$ distinct repetitions of the system, the resultant system has $2^m$ possible states. The equiprobability of the terminal states follows from each of the $m$ repetitions being independent fair coin-flips. We assemble $m$  $\BO(1)$ space coin-flips, so the total space is $\BO(m)$.
\hfill \qed
\end{proof}

\subsection{Robust Random Generator for General $n$}
We generalize this result to the problem of robust random $n$ generation for general $n$. We apply Corollary~\ref{cor:2s} to achieve two results. This first result, Theorem~\ref{thm:1nrandom}, achieves $0$ bias and $\mathcal{O}(\log n)$ expected space, but uses unbounded space in the worst case. Then, we consider a given space bound $s$ to construct a robust random $n$ generator with bias $\frac{1}{2^{\Theta(\frac{s}{\log n})}}$ guaranteed to use at most $s$ space in Theorem~\ref{thm:1nrandomBounded}.

\begin{theorem}\label{thm:1nrandom}
For $n\in \mathbb{Z^+}$, there exists a $\BO(\log n)$ expected space $2$-extensible fair robust random $n$ generator $\Gamma=(T,\sigma,2)$ in the aTAM with $|\sigma| = 1$ and $|T| = \BO(\log n)$.
\end{theorem}
\begin{proof}
According to Corollary~\ref{cor:2s} we can construct a robust random $2^m$ generator. Let $m$ be the smallest integer such that $n \leq 2^m$. We implement the robust random $2^m$ generator. We enumerate the $2^m$ terminal states of the system $0, 1,\dots,2^{m}-1$. We can then use additional tiles to read the final state of the $2^m$ generator. If the result exceeds $n-1$, we begin another $2^m$ generator and repeat the process. If the result does not exceed $n-1$, then we let the system terminate. We partition the set of terminal states by the result of the final $2^m$ generator, which must have assembled a state between $[0,n-1]$, or else it would have assembled an additional system, resulting in a robust random $n$ generator. The algorithm has $\BO(1)$ expected repetitions, and each repetition constructs a $\BO(m) = \BO(\log n)$ space system, resulting in $\BO(\log n)$ expected space.
\hfill \qed
\end{proof}

\begin{theorem}\label{thm:1nrandomBounded}
Given $n,s\in\mathbb{N}$ such that $s > 25\lceil\log n\rceil$, there exists a $2$-extensible robust random $n$ generator $\Gamma=(T,\sigma,2)$ with bias $\frac{1}{2^{\Theta \left (\frac{s}{\log n} \right )}}$ in the aTAM where the space of $\Gamma \leq s$, $|\sigma| = 1$, and $|T|=\BO(s + \log n)$.
\end{theorem}

\begin{proof}
We construct a similar system to that of Theorem~\ref{thm:1nrandom}. Let $m$ be the smallest integer such that $n \leq 2^m$. We implement the algorithm used in Theorem~\ref{thm:1nrandom}, except that we bound the number of repetitions of bit string generation to $k$ repetitions in order to meet the given space constraint.  The generation of each bit requires $15$ tiles as seen in Theorem~\ref{2extpositivesingleseed}. A set of tiles which grow a $5m\times1$ row on top of the generated bits are used to read the output to check if the result is $\leq n-1$. A set of $4k$ distinct tiles are used to create a $1\times4k$ column to bound the number of repetitions of the algorithm; that is, the algorithm will repeat a generation of $m$ bits if the output of the previous generation exceeds $n-1$ by introducing a cooperative binding location to continue attachment in the $1\times4k$ column.
With this method, a terminal assembly $\alpha$ of the system which uses $k$ repetitions satisfies $\lvert \alpha \rvert \leq 20mk+4k$, therefore a system implementing $k = \Theta (\frac{s}{m}) = \Theta (\frac{s}{\log n})$ satisfies $max \left \{ \left | \alpha \right | : \alpha \in TERM_\Gamma \right \} \leq s$.


Let $P_s(k)$ be the probability that the result is in the range $[0,n-1]$ in at most $k$ successive repetitions and $P_f(k) = 1 - P_s(k)$, which is the probability that $k$ successive repetitions generate a number in $[n,2^m-1]$.  The probability that the result of a single repetition is in $[0,n-1]$ (i.e., the system succeeds in generating a result in $[0,n-1]$ uniformly) is at least  $\frac{1}{2}$ since $2^m - n \leq n$. Therefore, the probability of failure after $k$ rounds is $P_f(k) \leq \frac{1}{2^k}$ and $P_s(k) \geq 1- \frac{1}{2^k}$. 
We partition the set of terminal states by the result of the final round.
If the result $r$ of the final repetition exceeds $n-1$, then we map the result by the function $f(r) = r-n$. That is, if the result $r$ of number generation in the final repetition is outside the range $[0,n-1]$, the system is considered to generate $r-n$. This maps $[n,2^m-1]$ to $[0, 2^m-n-1]$.
Therefore, in the final repetition, results in the range $[0,2^m-n-1]$ have a higher probability than results in the range $[2^m-n, n-1]$, resulting in bias of the generator.
The probability of generating a single number in $[2^m-n, n-1]$ given $k$ repetitions is $\frac{1}{n}P_s(k)$.
The probability of generating a single number in $[0, 2^m-n-1]$ is
$\frac{1}{n}P_s(k) + \frac{1}{2^m-n}P_f(k)$.
Then the bias of the generator is
\begin{align*}
\left |  \left ( \frac{1}{n}P_s(k) + \frac{1}{2^m-n}P_f(k) \right ) - \frac{1}{n}P_s(k) \right | &= \left |   \frac{1}{2^m-n}P_f(k) \right | \\
&\leq \left |  P_f(k) \right | \\
&\leq \frac{1}{2^k},
\end{align*}
therefore this robust random $n$ generator implementing $k = \Theta(\frac{s}{\log n})$ repetitions has bias $\frac{1}{2^{\Theta \left ( \frac{s}{\log n} \right )}}$.
\hfill \qed
\end{proof}

\section{1-Extensible Robust Fair Coin Flipping in the aTAM}\label{sec:1ext}
Thus far we have covered solutions to truly random coin flips that are robust to arbitrary concentrations.  We then used these gadgets to build expected length $n$ constant width lines at temperature two and one.  We also looked at approaching robust random number generation in the aTAM, and gave some methods for this that work in $\BO(\log n)$ expected space.

All of the positive results have required 2-extensibility, meaning in the assembly process there is at least one time when there are two possible locations tiles could be attach to. Following, we begin to look at some of the limitations of the aTAM with respect to randomness.  We then give some positive results despite the inherent limitations.  First we look at 1-extensible systems, and then at unstable concentrations (where the concentrations of tiles may change during assembly).

\subsection{$\BO(1)$ Space 1-extensible Robust Fair Coin Flipping}
A natural question from the $2$-extensible solution to the robust fair coin flip problem is whether there is also a $1$-extensible solution. We first answer this in the negative with Theorem~\ref{1-ext negative}, saying there is no $\mathcal{O}(1)$ space solution in the aTAM. However, using algorithms based on John von Neumann's randomness extractor \cite{von195113} we can achieve an unbounded space robust fair coin flip system (Theorem~\ref{1-ext positive}) as well as an $s$ space construction which incurs a small bias (Theorem~\ref{1-ext approximate positive}), for some space constraint $s$.  These are covered in Sections \ref{subsec:unb1ext} and \ref{subsec:fixed1ext}, respectively.

\begin{theorem}\label{1-ext negative}
There does not exist a $\BO(1)$ space 1-extensible robust fair coin flip tile system in the aTAM.
\end{theorem}

\begin{proof}
We prove this by contradiction. Assume that there exists a $\BO(1)$
space 1-extensible robust fair coin flip aTAM tile system
$\Gamma=(T,\sigma,\tau)$. We now specify a concentration distribution for $m$ tiles in $T$ that contradicts this claim.  Assume that $\Gamma$
generates assemblies of size at most $h$. Consider a series of
assemblies $p_1,\dots, p_n$ such that $p_{i+1}$ is derived from $p_{i}$
by the attachment of the tile in the frontier of $p_{i}$ with the
largest concentration. Select a parameter $t=10mn^3$, and let $c_1=1$ and $c_{i+1}=tc_i$ for $i=1,\dots, m-1$.  Let the concentration for each $t_i \in T$ be ${c_i\over c_1+c_2+\cdots+c_m}$.



For each assembly $p_i$, let $q_{i_1},\dots, q_{i_u}$ be the set of tile types in the frontier of ${p_i}$ listed in increasing order by their concentrations. Let $c_{i_u}$ denote the concentration of tile type $q_{i_u}$. With
probability ${c_{i_u}\over c_{i_1}+\cdots +c_{i_u}}$, tile type $q_{i_u}$
is attached. We have

\begin{align*}
{c_{i_u}\over c_{i_1}+\cdots +c_{i_u}} &\ge {1\over{(u-1)c_{i_{u-1}}\over c_{i_u}}+1}\\
&\ge {\frac{1}{{(u-1)\over t}+1}} \\
&\ge {\frac{1}{{m\over t}+1}}\\
&\ge {1\over {1\over 10n^3}+1}.
\end{align*}

Therefore, with probability at least

\begin{eqnarray*}
\left ({1\over {1\over 10n^3}+1} \right )^n&\ge& \left ({1\over {1\over
10n^3}+1}\right)^{10n^3\cdot {1\over 10n^2}}\\
&\ge&\left ({1\over e} \right )^{1\over 10n^2}\\
&>&0.6
\end{eqnarray*}

we follow
the sequence $p_1,\dots, p_n$ to generate an assembly. This is a
contradiction. Note that we use the facts that $(1+{1\over x})^x$ is
an increasing function for all real $x>1$, and $\lim_{x\rightarrow
+\infty}(1+{1\over x})^x=e\approx 2.17828$.
\hfill \qed
\end{proof}


\subsection{Unbounded Space, 1-Extensible, Robust Coin Flipping} \label{subsec:unb1ext}
In response to Theorem~\ref{1-ext negative}, which says there does not exist a $\BO(1)$ space 1-extensible robust fair coin flip tile system in the aTAM, we give a $1$-extensible aTAM system capable of robust fair coin flips in unbounded space. In 1951, John von Neumann gave a simple method for extracting a fair coin from a biased one \cite{von195113}. We show an algorithm based on the Von Neumann extractor. Algorithm \ref{alg:1-ext-unbounded-algorithm} uses an unbounded number of \emph{rounds} to extract a fair coin flip. We use Algorithm \ref{alg:1-ext-unbounded-algorithm} to show that a \emph{fair coin flip extractor} can be implemented in the aTAM to achieve an unbounded space, 1-extensible, robust coin flip tile system. Let Algorithm 2 denote an extension of this method in which we create a \emph{bounded fair coin flip extractor} by adding a parameter $k$ which controls the maximum number of rounds allowed. If after all $k$ rounds have been exhausted the system has not returned a fair coin flip, the result of a single flip is returned. We implement this \emph{bounded coin flip extractor} in the aTAM and achieve an $s$ space, 1-extensible, and robust coin flip tile system with bounded bias, for some space constraint $s$.

\begin{figure}[t]
\begin{minipage}{.5\textwidth}
\captionof{algorithm}{Unbounded}\label{alg:1-ext-unbounded-algorithm}
\begin{algorithmic}[1]
    \Require $h,t \in \mathbb{R}$, where $0<h<1, t=1-h$
    \Ensure \emph{heads} or \emph{tails}
    \Procedure{UnboundedFCFE}{$h,t$}
        \State $coin = \{heads, tails\}$
        \State $pdist = \{h, t\}$
        \Repeat
        \State $flip\_1 \gets flip(coin, pdist)$
        \State $flip\_2 \gets flip(coin, pdist)$
        \Until{$flip\_1$ $\neq flip\_2$}
        \State \textbf{return} $flip\_2$
    \EndProcedure
\end{algorithmic}
\end{minipage}%
\begin{minipage}[H]{.5\textwidth}
  \centering
  \centerline{\includegraphics[scale=.27]{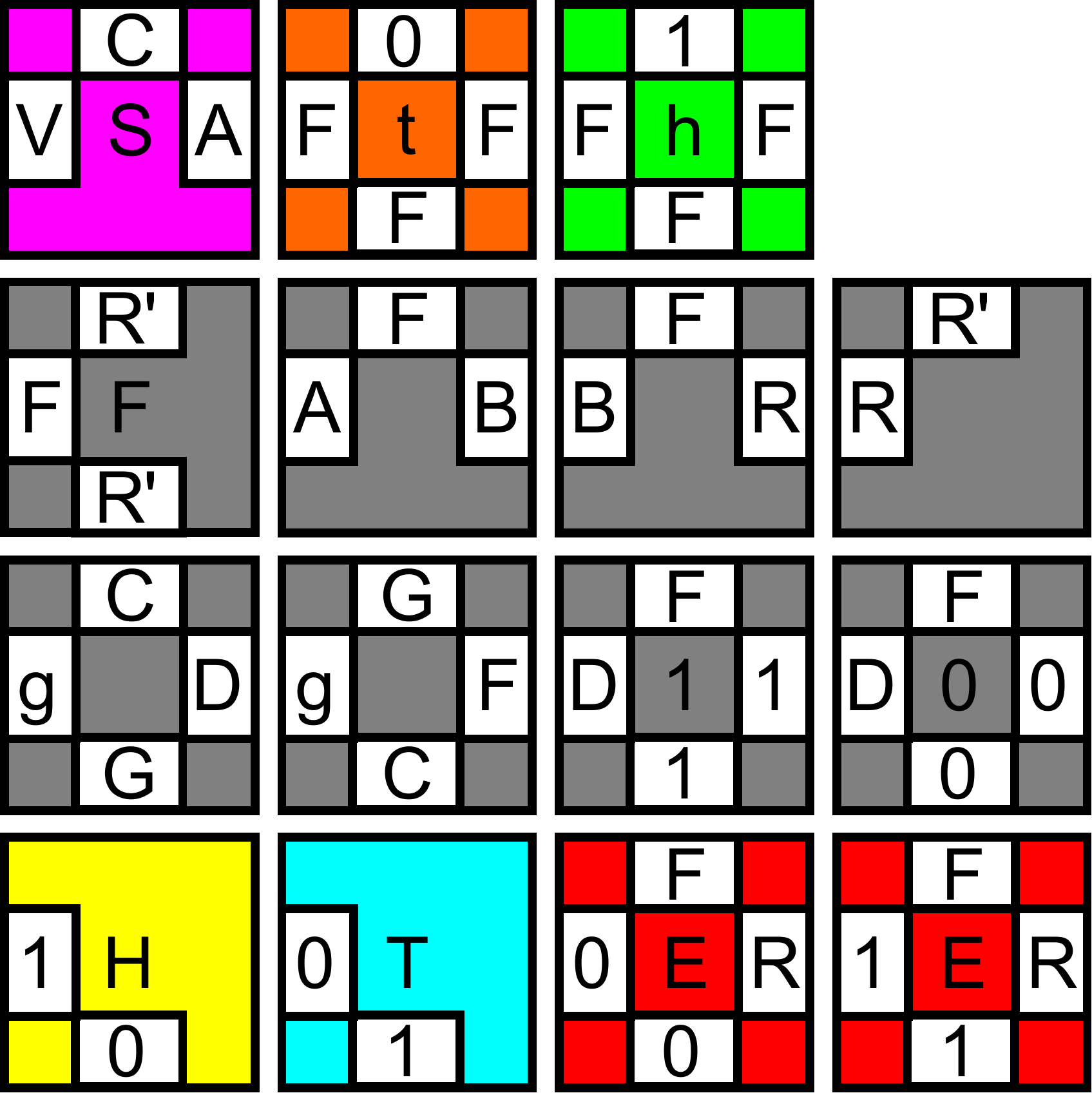}}
\end{minipage}

  \caption{Here we have the algorithm for an unbounded fair coin flip extractor on the left and the tile set for the construction that implements that algorithm on the right. The relative concentrations of the \emph{h tile} and \emph{t tile} serve as parameters $h$ and $t$, respectively. The tile labeled S is the seed of the tile assembly system and the temperature is 2. The strength of the glues are as follows: str(0)=1, str(1)=1, str(A)=2, str(B)=2, str(C)=1, str(D)=1, str(F)=1, str(G)=2, str(R)=2, and str(R')=2.}
  \label{fig:1-ext-pos-construction}
\end{figure}

We now describe our 1-extensible aTAM tile system that implements Algorithm \ref{alg:1-ext-unbounded-algorithm}. In Algorithm \ref{alg:1-ext-unbounded-algorithm}, a coin is a set of cardinality 2 with possible values $heads$ and $tails$. \emph{flip} is a function that selects and returns a \emph{heads} or \emph{tails} value based on the probabilities $h$ and $t$, respectively, where $h,t \in (0,1)$ and $h + t = 1$. In our construction, calls to the \emph{flip} function are carried out by a nondeterministic competition for attachment between a \emph{h tile} and a \emph{t tile}. Aside from calls to the \emph{flip} function, the rest of the algorithm can be implemented by deterministic tile placements. Figure \ref{fig:1-ext-pos-construction} gives the tile set used in the construction. This construction yields Theorem~\ref{1-ext positive}, and an example is shown in Figure \ref{fig:1-ext-pos-choices}.

\begin{figure}[ht]
    \centering
    \begin{subfigure}[b]{0.27\textwidth}
        \centering
        \centerline{\includegraphics[scale=.18]{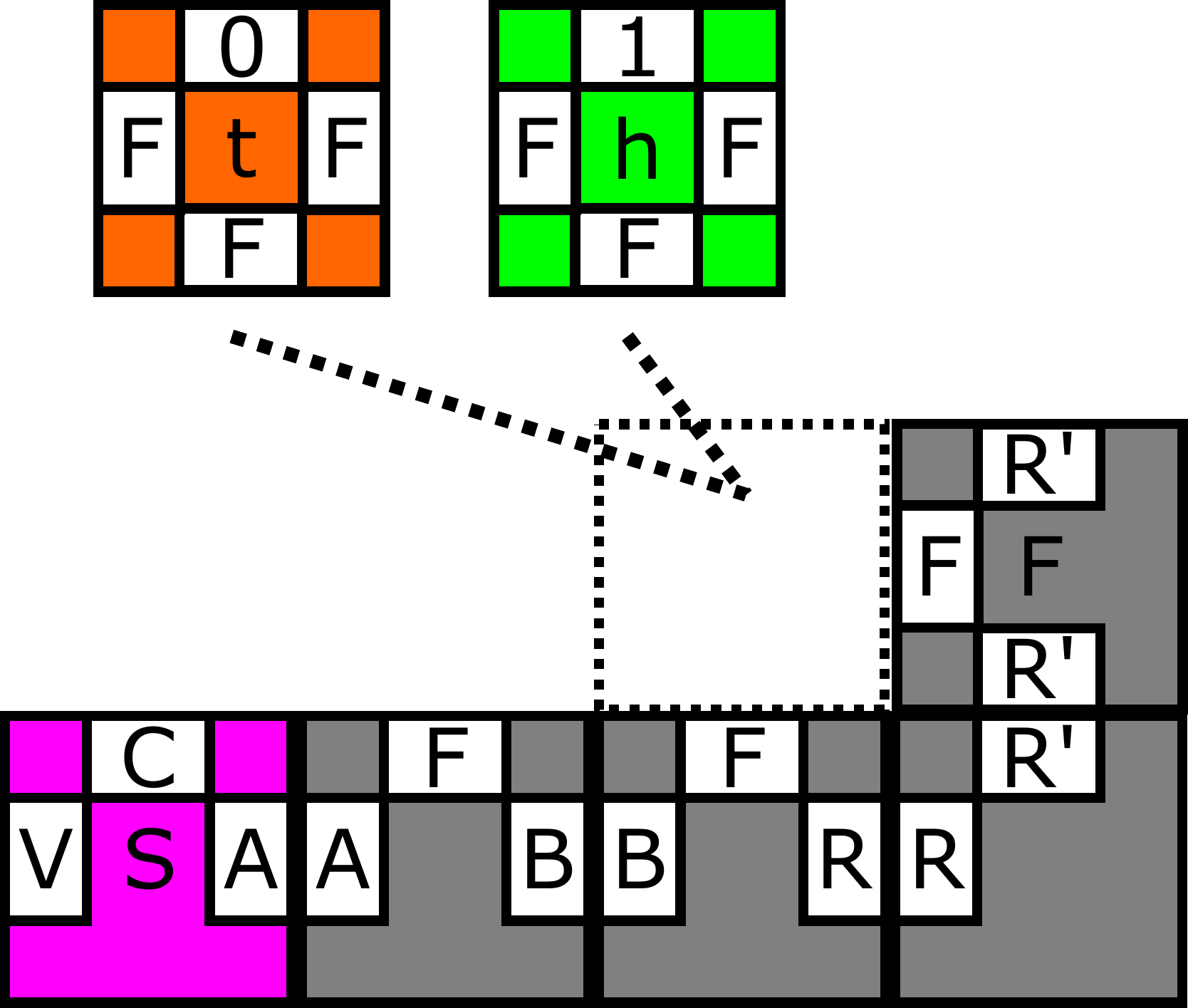}}
        \caption{An assembly with two possible choices for the next attachment corresponding to the first flip in the algorithm.}
    \end{subfigure}
    \hspace{.7cm}
    \begin{subfigure}[b]{0.27\textwidth}
        \centering
        \centerline{\includegraphics[scale=.18]{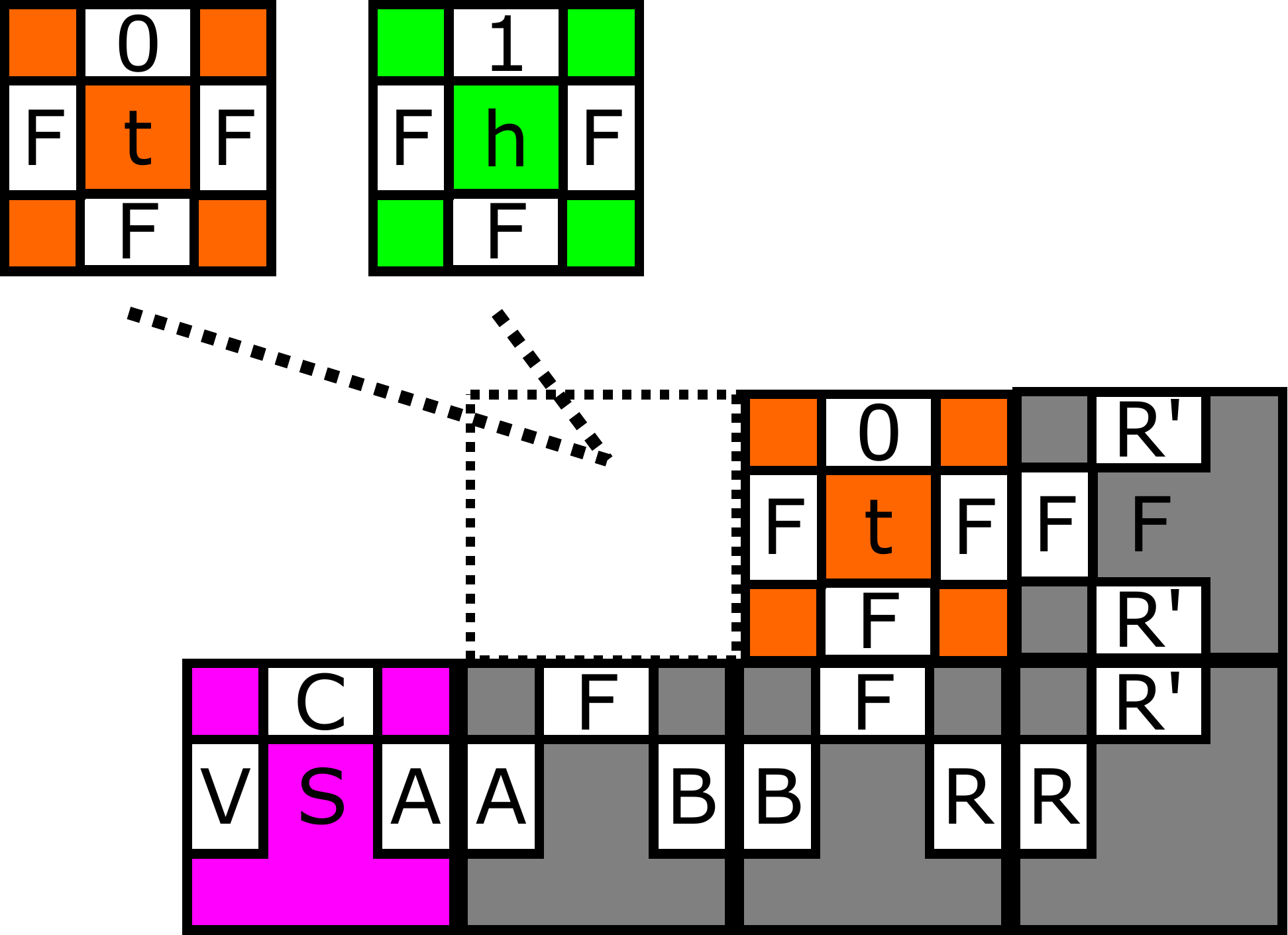}}
        \caption{Without loss of generality, this shows possible choices for the second flip of the algorithm after the first has been chosen.}
    \end{subfigure}
    \hspace{.7cm}
    \begin{subfigure}[b]{0.27\textwidth}
        \centering
        \centerline{\includegraphics[scale=.18]{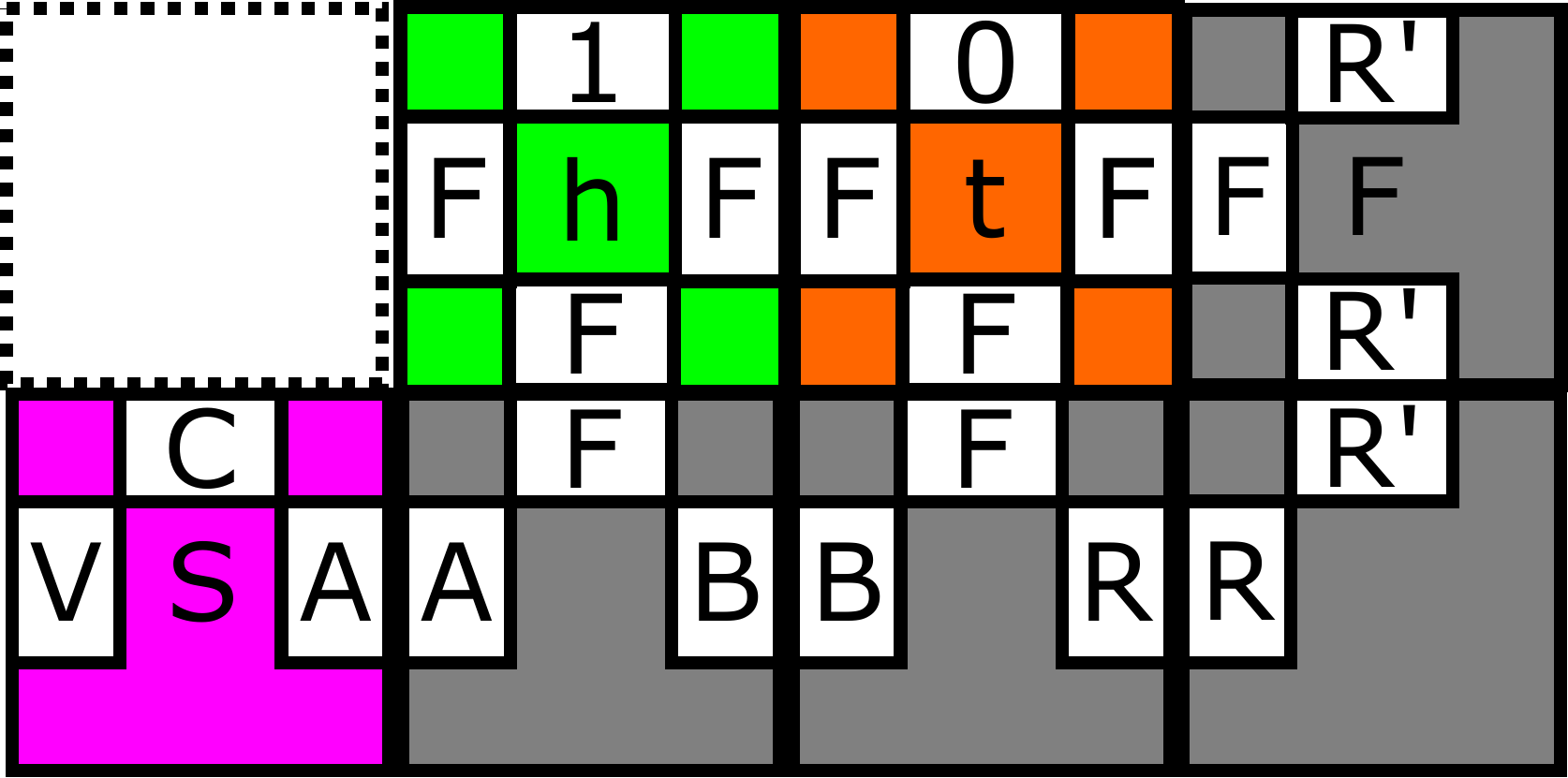}}
        \caption{A \emph{t tile} and a \emph{h tile} have been placed for the first and second flip, respectively. From Algorithm \ref{alg:1-ext-unbounded-algorithm}, this will return a \emph{heads}.}
    \end{subfigure}
    \begin{subfigure}[b]{0.27\textwidth}
        \centering
        \centerline{\includegraphics[scale=.18]{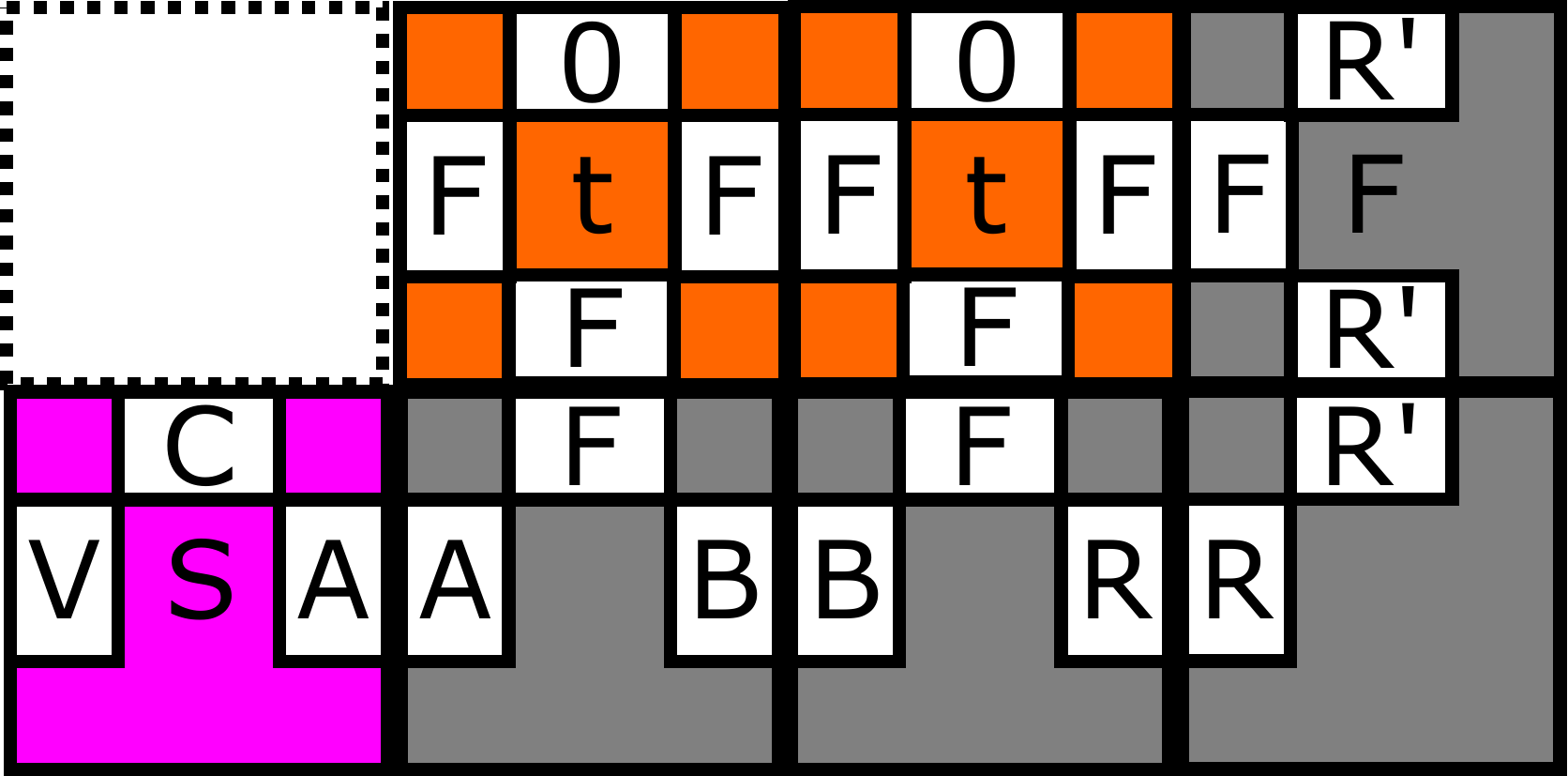}}
        \caption{Two \emph{t tiles} were placed for the first two flips. From Algorithm \ref{alg:1-ext-unbounded-algorithm}, the system must perform another round.}
    \end{subfigure}
    \hspace{.7cm}
    \begin{subfigure}[b]{0.27\textwidth}
        \centering
        \centerline{\includegraphics[scale=.18]{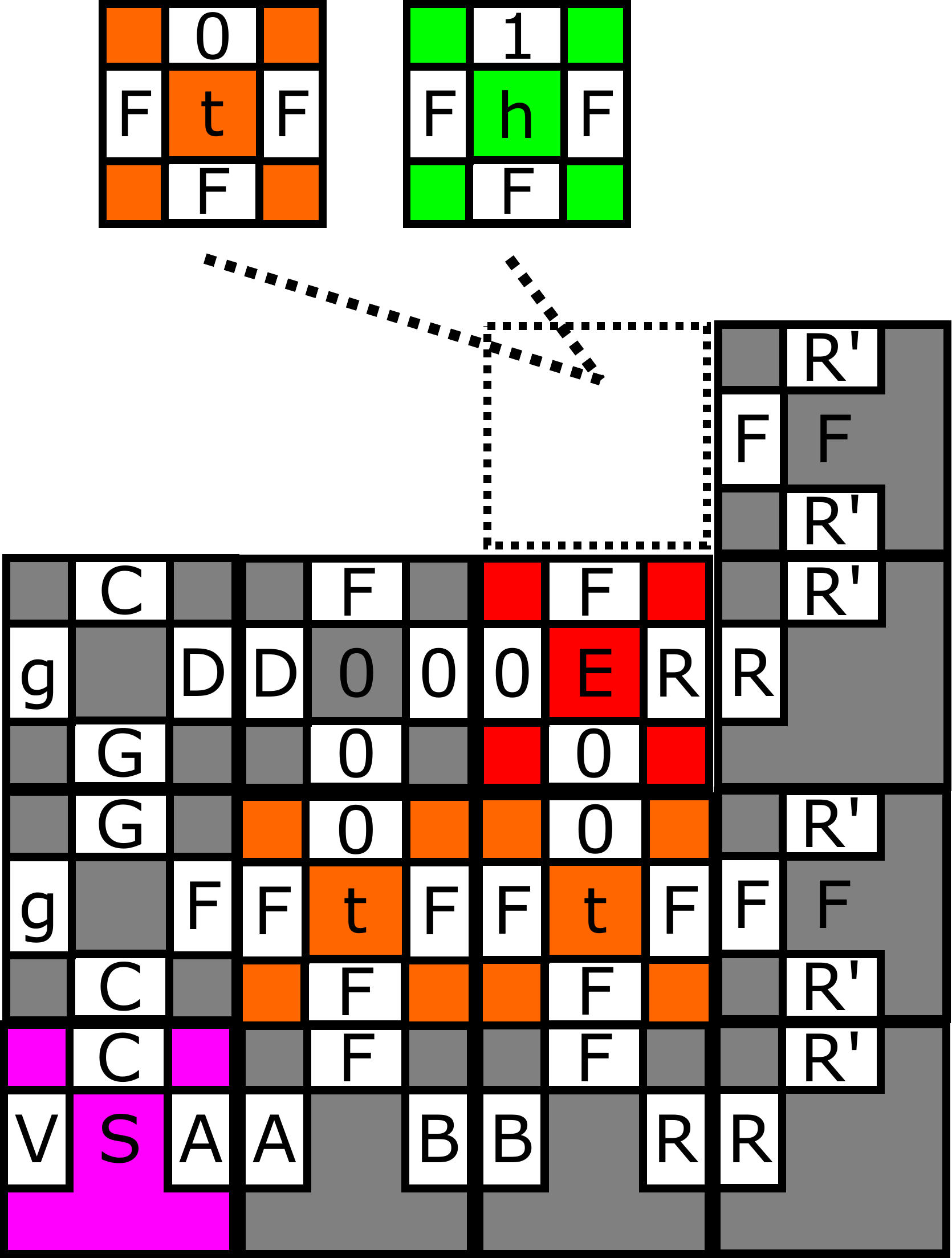}}
        \caption{An assembly where the first round of the algorithm failed to generate a bit and proceeds to start a new round.}
    \end{subfigure}
    \hspace{.7cm}
    \begin{subfigure}[b]{0.27\textwidth}
        \centering
        \centerline{\includegraphics[scale=.18]{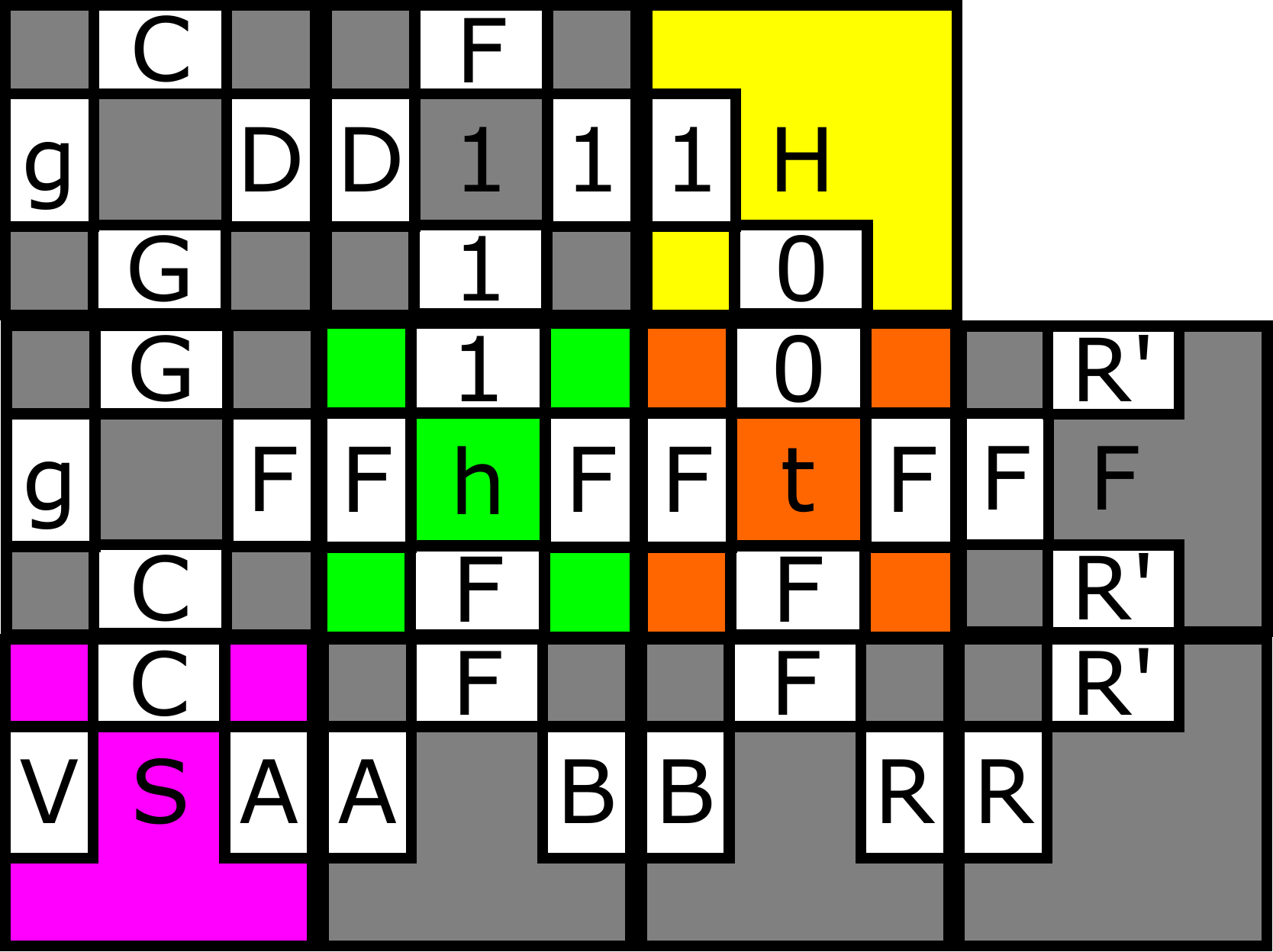}}
        \caption{An assembly where the first round of the algorithm was a valid flip and it generates a \emph{heads}.}
    \end{subfigure}
    \caption{A sample of producible assemblies for Round 1}
    \label{fig:1-ext-pos-choices}
\end{figure}

\begin{theorem}\label{1-ext positive}
There exists a 1-extensible tile system $\Gamma= (T,\sigma, 2)$ in the aTAM that implements a robust fair coin flip tile system (unbounded) and achieves $\mathcal{O} \left( \frac{1}{pq} \right)$ expected space, where $p$ and $q$ denote the relative concentrations of the two tiles with the largest difference in concentration for a given concentration distribution.
\end{theorem}

\begin{proof}
Let the probability that flipping a single coin is heads be represented by $u$ and tails be represented by $1-u$. In our construction, we have a \emph{h tile} and a \emph{t tile} with concentrations $C_{h}$ and $C_{t}$, respectively. Let $u=\frac{C_{h}}{C_{h}+C_{t}}$ and $v=1-u$. In each round, we flip two times. Let the probability of generating a bit each round be $g=2uv$. Then, let the expected total number of flips be $t$. If we succeed in the first round, we have only flipped twice. Otherwise, we have to start over, so the expected remaining number of flips is still two. Therefore, $t = 2g + (1-g)(2+t) = \frac{2}{g}$.

Using this strategy, each round requires two flips. Heads and tails each have a probability $uv$ of being generated. Thus, each round can succeed with a probability $2uv$ and the average number of flips required to generate a bit is $\frac{2}{g}=\frac{2}{2uv}=\frac{1}{uv}$. Since each round utilizes two flips, the expected number of rounds is then $\frac{t}{2}=\frac{1}{2uv}.$ In the best case, $u=v$ and the expected number of rounds would be $\frac{1}{2uv}=2$. In the worst case, the two tiles with the largest difference in concentration are the \emph{h tile} and \emph{t tile} implying $\frac{1}{2uv}=\frac{1}{2pq}$. Each round places a constant number of tiles $z$, therefore the expected space of generating a coin flip is the expected number of rounds multiplied by the number of tiles per round, $\frac{z}{2pq} = \mathcal{O} \left( \frac{1}{pq} \right)$. The placement of an \emph{H tile} or \emph{T tile} maps to the event that the algorithm returns \emph{heads} or \emph{tails}, respectively.
\hfill\qed
\end{proof}

\subsection{Fixed Space, 1-Extensible, Robust Coin Flipping} \label{subsec:fixed1ext}
We can also now limit the number of rounds, $k$, so that space of the system does not exceed some constraint $s$ by using some additional tile types and modification to glue strengths. If after $k$ rounds, the system has not returned a fair coin flip, the system returns the result of a single additional flip of the two tiles used in nondeterministic attachment. This \emph{bounded fair coin flip extractor} can be implemented in the aTAM to achieve a fixed space, 1-extensible, robust coin flip tile system with bounded bias.
The bounded $k$-rounds can be controlled by first constructing a column of height $\BO(k)$ with glues that allow the variant of the construction of Theorem \ref{1-ext positive} to grow along the right edge of the column. Note that this column can be built more efficiently, by allowing some width, using a 1-extensible version of the aTAM counter construction from \cite{AGKS05g} for a desired base, leading to a tradeoff in bias, space, and tile complexity.

\begin{theorem}\label{1-ext approximate positive}
There exists an $s$ space 1-extensible robust coin flip tile system in the aTAM with bias $p^{(s/10)}$, where $p$ is the larger relative concentration from the pair of tiles with the largest difference in concentration for a given concentration distribution.
\end{theorem}

\begin{proof}
We need at most 7 tile placements to perform a single additional flip when we fail all allotted rounds and each round places 10 tiles. We design a system that can perform as many possible rounds, $k$, given $s-7$ space, where

\begin{equation}\label{eqn:compute_k_given_c}
\frac{s}{10}-1 < k = \lfloor (s-7)/10 \rfloor < \frac{s}{10}.
\end{equation}

In the worst case, the two tile types with the largest difference in concentration, for a given concentration distribution, are the two tile types used in nondeterministic attachment. In our construction, those tiles are the \emph{h tile} and a \emph{t tile} with concentrations $C_{h}$ and $C_{t}$, respectively. Without loss of generality, consider that $C_h > C_t$ and thus, $p = \frac{C_h}{C_h + C_t}$ and let $q=1-p$. Let $P(X=heads)$ denote the probability that this system returns a heads and $P(X=tails)=1-P(X=heads)$. Let $F_k$ denote the probability that the system fails to return a coin flip after $k$ rounds, that is $F_k = (1-2pq)^k$. Therefore, $P(X=heads)=pF_k + \frac{1-F_k}{2}$ and

\begin{equation}
\frac{|P(X=heads) - P(X=tails)|}{2} = F_k \left( p-\frac{1}{2} \right).
\end{equation}

And, since $\frac{1}{2} < p < 1$ and $p^k \ge F_k$,

\begin{align}
\begin{split}
F_k \left( p-\frac{1}{2} \right) &\le p^k \left(p-\frac{1}{2} \right) \\
&\le p^{k+1} \\
&< p^{\frac{s}{10}}.
\end{split}
\end{align}

Therefore,

\begin{equation}
2p^{\frac{s}{10}} \ge |P(X=heads) - P(X=tails)|,
\end{equation}

which implies this system has bias $p^{\frac{s}{10}}$.
\hfill\qed
\end{proof}

\section{Robust Fair Coins with Unstable Concentrations}\label{sec:unstable}
As an extension to the idea of concentration independent solutions outlined in this paper, we consider an adversarial model wherein the concentration distribution of tiles changes during each stage of the assembly process; in other words, the concentrations are unstable.

\begin{definition}[Unstable Concentrations Robust Fair Coin Flip] \label{def:unstable_robust}
Let an \emph{unstable concentration distribution} $P$ be a function mapping $z \in \mathbb{Z^+}$ to concentration distributions over a tile set $T$.  Let $P_i$ denote $P(i)$. For each $B$ that satisfies $A \rightarrow^\Gamma_1 B$, let $t_{A\to B}$ denote the tile in $T$ whose translation is added to $A$ to get $B$. The transition probability from $A$ to $B$ is defined to be

	\begin{equation*}
	TRANS(A,B) = \dfrac{P_{|A|}(t_{A\to B})}{\sum\nolimits_{\{C \vert A\rightarrow^\Gamma_1 C\}} P_{|A|}(t_{A\to C})}
	\end{equation*}


 A finite tile system $\Gamma$ an \textbf{unstable concentrations robust fair coin flip} iff the set of terminal assemblies in $\prodset{\Gamma}$ is partitionable into two sets $X$ and $Y$ such that $\sum\limits_{x \in X} \probterm{\Gamma \rightarrow x}{C} = \sum\limits_{y \in Y} \probterm{\Gamma \rightarrow y}{C}$ for all unstable concentration distributions $C$.

\end{definition} 

We now prove that there is no unstable concentration robust fair coin flip system in the aTAM.
First, we state and prove a lemma that will be useful in our proof.  Informally, by attaching the same tile type repeatedly to an assembly until it is no longer possible to attach the tile, the resulting assembly will be unique.

\begin{lemma}\label{same-lemma}
For any producible assembly $A_0 \in \prodset{\Gamma}$ and any tile type $t \in T$,  there exists another assembly
$A^*$ such that for any sequence of assemblies  $\langle A_0,A_1,
A_2,\ldots,A_h\rangle$ where tile type $t$ can not be attached to $A_h$, and each $A_{i+1}$ is derived from $A_i$ $(i=0,1,2,\cdots, h-1)$ by attaching a tile of type $t$, then $A_h=A^*$.
\end{lemma}

\begin{proof}
Let $A^*$ be the least-sized producible assembly such that $A^*\setminus A_0$ contains only tiles of type $t$ and the frontier of $A^*$ contains no tiles of type $t$. We will show that $A_0$ can only grow $A^*$ if only allowed to attach tile type $t$.

Towards a contradiction, assume there exists a sequence of assemblies from $A_0$ such that $A_h\neq A^*$. If $A_h$ is some subassembly of $A^*$, note that we may still attach tiles of type $t$ to reach $A^*$, implying that $A_h$ does not fit the specified requirements. Otherwise, let $A_n$ be the first assembly in the sequence which contains a tile not in $A^*$. Consider $A_{n-1}$. There is no tile of type $t$ attachable to $A_{n-1}$ such that the tile is not in $A^*$. If there were, that tile of type $t$ would be attachable to $A^*$, contradicting the definition of $A^*$.  Therefore no such $A_n$ can exist, implying that $A_h = A^*$. \hfill \qed
\end{proof}






\begin{theorem}\label{thm:unstable_neg}
There does not exist a $\BO(1)$ space unstable concentrations robust fair coin flip tile system in the aTAM.
\end{theorem}

\begin{proof}
Towards a contradiction, assume that a space-$n$ solution does exist.

As the assembly process proceeds, the key point to consider is when the current assembly enters a state in which multiple distinct
positions may attach a tile. In such a case select one type $t$ of all attachable tiles, and increase
its concentration to ensure, with high probability, that assembly proceeds by attaching only tiles of type t up until there is no position to attach type $t$
tiles.           Such a type $t$ is called a \emph{dominate} type. Let the
concentration of the dominate tile type $t$ be $(1-{1\over
100n^2})$. For each step $i$, let $t_i$ denote the dominate type of
concentration $(1-{1\over 100n^2})$.

When there is more than one position to attach the same type of tile $t$, we are assured by Lemma~\ref{same-lemma} that a unique assembly will result after repeatedly placing tiles of type t (in any order) until placement of t is no longer an option.

Given this setup, we have that at each step $i$, the assembly does not grow with a dominate type
with probability at most ${1\over 10n^2}$. With probability at most
${1\over 10n}$, there is a step $i$ among $n$ steps that the
assembly does not grow with the dominate type.

Therefore, there is a terminal assembly that will be generated with probability at least
$0.9$. This is a contradiction. \hfill \qed
\end{proof}

\section{Other Self-Assembly Models}\label{sec:alternative}

\begin{figure}[h!]
    \centering
    \begin{subfigure}[b]{0.2\textwidth}
        \centering
        \includegraphics[scale=.40]{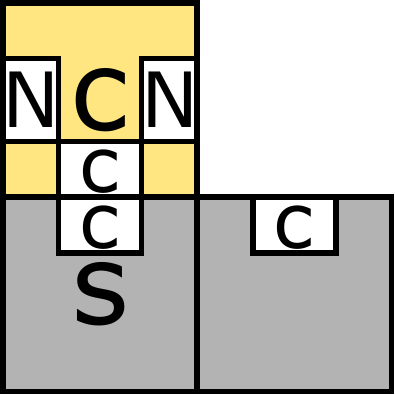}
        \caption{}
        \label{fig:neg_atam_heads}
    \end{subfigure}
    \begin{subfigure}[b]{0.2\textwidth}
        \centering
        \includegraphics[scale=.40]{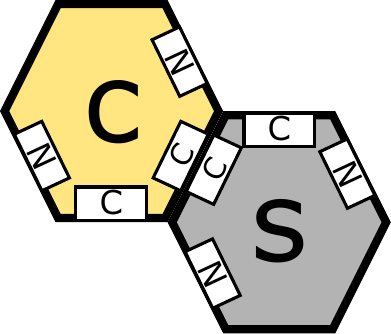}
        \caption{}
        \label{fig:hex_heads}
    \end{subfigure}
    \begin{subfigure}[b]{0.2\textwidth}
        \centering
        \includegraphics[scale=.40]{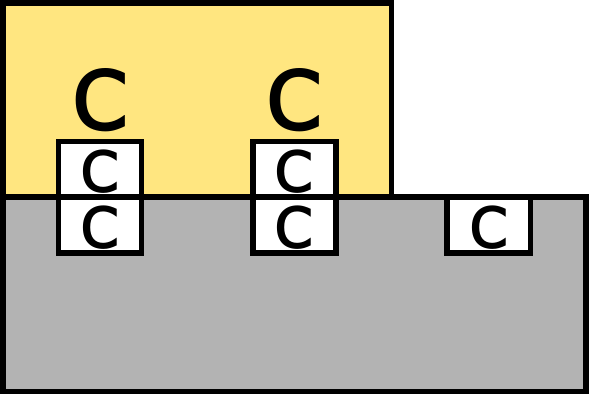}
        \caption{}
        \label{fig:polytam_heads}
    \end{subfigure}
    \begin{subfigure}[b]{0.2\textwidth}
        \centering
        \includegraphics[scale=.40]{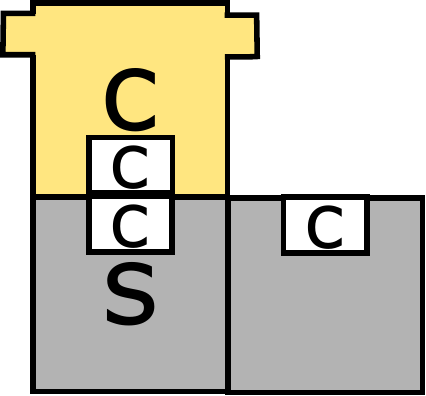}
        \caption{}
        \label{fig:gtam_heads}
    \end{subfigure}
    \caption{The terminal assemblies representing ``heads'' in some alternate models.
    $C$ is a strength-$\tau$ glue and $N$ is a strength-$(-1)$ glue in (a) the aTAM tile system and (b) the hexTAM tile system.
    (c) $C$ is a strength-$1$ glue in a $\tau = 2$ polyTAM tile system.
    (d) $C$ is a strength-$1$ glue in a $\tau = 1$ GTAM  tile system. The abutting geometry does not allow two $C$ tiles to attach.}
    \label{fig:alternate_heads}
\end{figure}

Motivated by the impossibility of robust coin flipping in the aTAM under unstable concentrations, we now consider some established extensions of the aTAM from the literature.  In particular, we show that robust coin flipping with unstable concentrations is possible within the aTAM with negative glues~\cite{DotKarMasNegativeJournal,rgTAM,Winf98},  the hexTAM~\cite{oneTile2014} with negative glues, the polyTAM~\cite{HPR2015UCA}, and the GTAM~\cite{fu2012SAGT}.

\begin{definition}[The Abstract Tile Assembly Model with Negative Interactions]
\label{def:atam_neg}
In the \emph{abstract Tile Assembly Model with Negative Interactions} \cite{Winf98,DotKarMasNegativeJournal,rgTAM}, the restriction that each glue type $g \in \Pi$ must be of non-negative integer strength is removed. We relax this requirement and allow any $g \in \Pi$ to have $str(g) \in \mathbb{Z}$.
\end{definition}

\begin{definition}[The Polyomino Tile Assembly Model]
\label{def:polytam}
In the \emph{Polyomino Tile Assembly Model} (polyTAM) \cite{HPR2015UCA}, a tile assembly system $\Gamma=(T, \sigma, \tau)$ is such that $T$ is the set of polyomino tiles. A polyomino tile can easily be thought of as an arrangement of aTAM tiles, where every tile is adjacent to at least one other tile. These adjacent tiles are bonded with an infinite strength. $\sigma$ is a $\tau$-stable assembly of polyomino tiles. $\tau$ is defined as for the abstract Tile Assembly Model.
\end{definition}

\begin{definition}[The Hexagonal Tile Assembly Model with Negative Interactions]
\label{def:hextam_neg}
In the \emph{Hexagonal Tile Assembly Model} (hexTAM) \cite{oneTile2014}, a tile assembly system $\Gamma=(T, \sigma, \tau)$ is such that each tile in $T$ is a regular unit hexagon. Similar to the aTAM with Negative Interactions Definition \ref{def:atam_neg}, there is no restriction that each glue type $g \in \Pi$ must be of non-negative integer. $\sigma$ and $\tau$ are defined as they are for the abstract Tile Assembly Model.
\end{definition}

\begin{definition}[The Geometric Tile Assembly Model]
\label{def:gtam}
In the \emph{Geometric Tile Assembly Model} (GTAM) \cite{fu2012SAGT}, a tile assembly system $\Gamma=(T, \sigma, \tau)$ is such that each edge of the tiles in $T$ are assigned a geometric pattern. Tile attachments that would result in an overlap of edge geometries are disallowed. $\sigma$ and $\tau$ are defined as for the abstract Tile Assembly Model.
\end{definition}

\begin{theorem}\label{thm:secondary_models}
There exists a $\BO(1)$ space unstable concentration robust fair coin-flip tile system in the aTAM with negative glues, polyTAM, hexTAM with negative glues, and the GTAM.
\end{theorem}

\begin{proof}\label{proof:secondary_models}
Consider a tile assembly system $\Gamma = (T,\sigma,\tau)$ with $3$ producible assemblies: $\sigma$, a terminal assembly $heads$, and a terminal assembly $tails$. Further, $\sigma \rightarrow^\Gamma_1 heads$ and $\sigma \rightarrow^\Gamma_1 tails$. Let $t_{\sigma \to heads}$ and $t_{\sigma \to tails}$ be the same tile $c$, then $TRANS(\sigma, heads) = TRANS(\sigma, tails) = \frac{P(c)}{2P(c)} = \frac{1}{2}$. Systems which meet these characteristics in the mentioned models are shown in Figure~\ref{fig:alternate_heads}.\hfill \qed
\end{proof}



\section{Conclusions and Future Work}\label{sec:conclusion}
In this paper we have introduced the problem of designing robust random number generating systems.  Generating such random numbers is fundamental for the implementation of randomized self-assembly algorithms.  By incorporating concentration independent robustness into the design of such systems, we directly address the practical issue of limited control over species concentrations.  Our goal in this work is to provide a stepping stone for the creation of general, robust randomized self-assembly systems.  As evidence towards the feasibility of this goal, we have shown how our gadgets can be applied to convert a large class of linear systems into equivalent systems with the concentration robustness property.  A more general open problem is as follows: given a general tile system, is it possible to convert the system to an approximately equivalent system that is concentration robust?  If possible, how efficiently can this be accomplished in terms of scale factor and approximation factor?

Another direction for future work is the consideration of generalizations of the coin flip problem.  Our partition definition for coin flip systems extends naturally to distributions with more than two outcomes, as well as non-uniform distributions.  What general probability distributions can be assembled in $\mathcal{O}(1)$ space, and with what efficiency?  We have also introduced the online variant of concentration robustness in which species concentrations may change at each step of the self-assembly process.  We have shown that when such changes are completely arbitrary, coin flipping is not possible in the aTAM.  A relaxed version of this robustness constraint could permit concentration changes to be bounded by some fixed rate.  In such a model, how close to a fair coin flip can a system guarantee in terms of the given rate bound?  As an additional relaxation, one could consider the problem in which an initial concentration assignment may be \emph{approximately} set by the system designer, thereby modeling the limited precision an experimenter can obtain with a pipette.

A final line of future work focuses on applying randomization in self-assembly to computing functions.  The parallelization within the abstract tile assembly model allows for substantially faster arithmetic than what is possible in non-parallel computational models~\cite{Keenan2016}.  Can randomization be applied to solve these problems even faster?  Moreover, there are a number of potentially interesting problems that might be helped by randomization, such as primality testing, sorting, or general simulation of randomized boolean circuits.


\bibliographystyle{plain}
\bibliography{tam}
\end{document}